\documentclass[preprint,12pt]{elsarticle}

\usepackage[utf8]{inputenc}
\usepackage[ruled,vlined,linesnumbered]{algorithm2e}
\usepackage{amsthm,booktabs}
\usepackage{color}
\usepackage{tikz,comment}
\usepackage{subfiles}
\usepackage{soul}

\newtheorem{theorem}{Theorem}

\newtheorem{lemma}[theorem]{Lemma}

\def \polylog{\operatorname{polylog}}
\usepackage{mathtools}

\DeclarePairedDelimiter\floor{\lfloor}{\rfloor}

\usepackage{array}
\newcolumntype{P}[1]{>{\centering\arraybackslash}p{#1}}
\usepackage{tabularx}

\usepackage{hyperref}
\usepackage{mathtools}

\usepackage{array}
\newcolumntype{M}[1]{>{\centering\arraybackslash}m{#1}}
\usepackage{multirow,booktabs}

\newcommand{\anis}[1]{{\color{blue}\underline{\textsf{Anisur:}}} {\color{blue} \emph{#1}}}
\newcommand{\manish}[1]{{\color{magenta}\underline{\textsf{Manish:}}} {\color{olive} \emph{#1}}}

\SetCommentSty{mycommfont}

\newcommand{\cR}{\mathcal{R}}

\journal{Arxiv}

\begin{document}

\begin{frontmatter}



\title{Agentic Distributed Computing
\tnoteref{t1}}
\tnotetext[t1]{The results for the case of $k=n$ were appeared as a brief announcement in the Proceedings of DISC 2024~\cite{KKMS24}.}



\author[inst1]{Ajay D. Kshemkalyani}\ead{ajay@uic.edu}
\author[inst2]{Manish Kumar\corref{mycorrespondingauthor}\fnref{fund2}}\ead{manishsky27@gmail.com}
\author[inst3]{Anisur Rahaman Molla\fnref{fund3}}\ead{molla@isical.ac.in}
\author[inst4]{Gokarna Sharma}\ead{gsharma2@kent.edu}

\address[inst1]{Department of Computer Science, 
University of Illinois Chicago, Chicago, IL 60607, USA}
\address[inst2]{Department of Computer Science \& Engineering, IIT Madras, Chennai 600036, India}
\address[inst3]{Cryptology and Security Research Unit, Indian Statistical Institute, Kolkata,  700108, India}
\address[inst4]{Department of Computer Science, Kent State University, Kent, OH 44242, USA}

\cortext[mycorrespondingauthor]{Corresponding author.}

\fntext[fund2]{The work of Manish Kumar is supported by CyStar at IIT Madras.}
\fntext[fund3]{The work of A. R. Molla was supported, in part, by ISI DCSW Project, file no. H5413.}
\begin{abstract}
The most celebrated and extensively studied model of distributed computing is the {\em message-passing model,} in which each vertex/node of the (distributed network) graph corresponds to a static computational device that communicates with other devices through passing messages. 
In this paper, we consider the {\em agentic model} of distributed computing which extends the message-passing model in a new direction. In the agentic model,  computational devices are modeled as relocatable or mobile computational devices (called agents in this paper), i.e., each vertex/node of the graph serves as a container for the devices, and hence communicating with another device requires relocating to the same node. 
We study two fundamental graph level tasks, leader election, and minimum spanning tree, in the agentic model, which will enhance our understanding of distributed computation across paradigms. 
The objective is to minimize both time and memory complexities. 
Following the literature, we consider the synchronous setting in which each agent performs its operations synchronously with others, and hence the time complexity can be measured in rounds. 
In this paper, we present two deterministic algorithms for leader election: one for the case of $k<n$ and another for the case of $k=n$, minimizing both time and memory complexities, where $k$ and $n$, respectively, are the number of agents and number of nodes of the graph.    
Using these leader election results, we develop deterministic algorithms for agents to construct a minimum spanning tree of the graph, minimizing both time and memory complexities. To the best of our knowledge, this is the first study of distributed graph level tasks in the agentic model with $k\leq n$. Previous studies only considered the case of $k=n$.
\end{abstract}


\end{frontmatter}

\section{Introduction}
\label{section:introduction}
The well-studied {\em message-passing model} of distributed computing assumes an underlying distributed network represented as an undirected graph $G=(V,E)$, where each vertex/node corresponds to a {\em computational device} (such as a computer or a processor), and each edge corresponds to a bi-directional communication link. Each node $v\in V$ has a distinct $\Theta(\log n)$-bit identifier, $n=|V|$. 
The structure of $G$ (topology, latency) is assumed to be unknown in advance, and each node typically knows only its neighboring nodes. 
The nodes interact with one another by sending messages (hence the name message-passing) to achieve a common goal. 
The computation proceeds according to synchronized {\em rounds}. In each round, each node $v$ can perform unlimited local computation and may send a distinct message to each of its neighbors. Additionally, each node $v$ is assumed to have no restriction on storage. 
In the {\sf LOCAL} variant of this model, there is no limit on bandwidth, i.e., a node can send any size message to each of its neighbors. In the {\sf CONGEST} variant, bandwidth is taken into account, i.e.,  a node may send only a, possibly distinct, $O(\log n)$-bit message to each of its neighbors.

In this paper,  we consider the {\em agentic} model of distributed computing which extends the message-passing model in a new direction. 
In the agentic model, the computational devices are modeled as  {\em relocatable or mobile computational devices} (called agents in this paper and hence the name `agentic' for the computing model). Departing from the notion of vertex/node as a {\em static} device in the message-passing model,  the vertices/nodes serve as {\em containers} for the devices in the agentic model. 
The agentic model has two major differences with the message-passing model (Table \ref{table:model-comparison} compares the properties). 

\begin{table}[h!]
\footnotesize{
\centering
\begin{tabular}{ccccc}
\toprule
{\bf Model} & {\bf Devices} & {\bf Local com-} & {\bf Device} & {\bf Neighbor} \\ 
& & {\bf putation} & {\bf storage} & {\bf communication} \\ 
\toprule
Message-passing & Static & Unlimited & Unrestricted& Messages \\ 
\hline
Agentic & Mobile & Unlimited & Limited & Relocation \\ 
\bottomrule
\end{tabular}
\caption{Comparison of the message-passing and agentic models. 
\label{table:model-comparison}
}
}
\end{table}

\begin{itemize}
\item [] {\bf Difference I.} The graph nodes neither have identifiers, computation ability, nor storage,  but the devices are assumed to have distinct $O(\log n)$-bit identifiers, computation ability, and limited storage. 
\item [] {\bf Difference II.} The devices cannot send messages to other devices except the ones co-located at the same node. To send a message to a device positioned at a neighboring node, a device needs to relocate to the neighbor. 
\end{itemize}
Difference II is the major problem for the agentic model. 
To complicate further, while a device relocates to a neighbor, the device at that neighbor might relocate to another neighbor. 
%
%
%
This was not an issue in the message-passing model since, when a (static) device sends messages, it is always the case that neighboring nodes receive those messages. 
In this paper, we study the fundamental graph level tasks, such as leader election and minimum spanning tree in the agentic model,  
which will enhance our understanding of distributed computation across paradigms.

The agentic model is useful, applicable, and preferred to the message-passing model in scenarios like private networks in the military or sensor networks in inaccessible terrain where direct access to the network is possibly obstructed, but small battery-powered relocatable computational devices can learn network structures and their properties for overall network management. Prominent use of the agentic model for network management is in underwater navigation \cite{Cong2021}, network-centric warfare in military systems \cite{LEE2018}, modeling social network \cite{Zhuge2018}, studying social epidemiology \cite{ELSAYED2012}, etc. 


First of all, with $k\leq n$ agents working on $G$, we develop two sets of deterministic algorithms for leader election, one for the case of $k<n$ and another fo the case of $k=n$,  with provable guarantees on two performance metrics that are fundamental to the agentic model: {\em time complexity} of a solution and {\em storage requirement} per agent. We focus on the {\em deterministic} algorithms since they may be more suitable for relocatable devices. 
%
Our consideration of storage complexity as the second performance metric is important as this metric was often neglected in the message-passing model with the implicit assumption that the devices have no restriction on the amount of storage needed to run the algorithm successfully. Instead, the focus was given on {\em message complexity} -- the total number of messages sent by all nodes for a solution \cite{Pandurangan0S18} as the second performance metric in the message-passing model. 

Using the proposed leader election algorithms, 
we 
give stabilizing and explicit algorithms to construct a minimum spanning tree (MST) of $G$, another fundamental and well-studied problem in the message-passing model of distributed computing, in the agentic model, and 
provide both time and memory complexities. 

To the best of our knowledge, this is the first time these two fundamental graph level tasks were studied in the agentic model for any $k\leq n$ in a systematic way, which will be more apparent in the following. To the best of our knowledge, no previous studies studied graph level tasks, including leader election and MST, for the case of $k<n$. Leader election and MST were studied for the case of $k=n$ \cite{
kshemkalyani2024agent,KshemkalyaniAAMAS25,KshemkalyaniICDCIT25}, which our work in this paper subsumes as we consider $k\leq n$.

\vspace{2mm}
\noindent{\bf Computing Model.}
We model the network as a connected, undirected graph $G=(V,E)$ with $|V|=n$ nodes and $|E|=m$ edges.  Each node $v_i \in V$ has $\delta_i$ ports corresponding to each edge incident to it labeled $[1, \ldots, \delta_i]$. For a node $w$, we denote its neighbor nodes in $G$ by $N(u)$. Consider an edge $(u,v)$ that connects node $u$ and node $v$ (and node $v$ to node $u$). We note the port at $u$ leading to $v$ by $p_{uv}$ and the port at $v$ leading to $u$ by $p_{vu}$.   
We assume that the set $\cR=\{r_1,r_2,\ldots,r_k\}$ of $k\leq n$  agents are initially positioned on the nodes of $G$. 
The agents have unique IDs in the range $[1,k^{O(1)}]$.

Initially, a node may have zero, one, or multiple agents positioned. An agent at a node can communicate with some (or all) agents present at that node, but not with the agents that are situated at some other nodes (i.e., we consider the {\em local communication} model \cite{Augustine:2018}). 

An agent can move from node $v$ to node $u$ along the edge $(v,u)$.
Following the message-passing literature, e.g., \cite{GarayKP93}, we assume that an agent can traverse an edge in a round, irrespective of its (edge) weight (we consider  $G$ weighted only for the MST problem in this paper; leader election does not require $G$ to be weighted). 
An agent that moves from node $v$ along the port $p_{vu}$ is aware of the port $p_{uv}$ when it arrives at $u$.
Additionally, at any node $v$, it is aware of the weight $w(e)$ (if $G$ is weighted) of the edge $(v,u)$ that connects $v$ to its neighbor node $u$. We assume that there is no correlation between the two port numbers of an edge. 

Any number of agents are allowed to move along an edge at any time; that is, the agentic model does not put restrictions on how many agents can traverse an edge at a time.

The agents use the synchronous setting as in the standard ${\sf CONGEST}$ model: In each round, each agent $r$ positioned at a node $v$ can perform some local computation based on the information available in its storage as well as the (unique) port labels at the positioned node and decide to either stay at that node or exit it through a port to reach a neighboring node. Before exiting, the agent might write information on the storage of another agent that is staying at that node.   An agent exiting a node always ends up reaching another node (i.e., an agent is never positioned on an edge at the end of a round).  

The time complexity is the number of rounds of operations until a solution. The storage (memory) complexity is the number of bits of information stored at each agent throughout the execution.  

Notice that at any round, the agent positions on $G$ may satisfy the following: 
\begin{itemize}
\item {\em dispersed} -- $k (\leq n)$ agents are on $k$ nodes,  

\item {\em general} -- not {\em dispersed} (i.e., at least two agents on at least one node). 
\end{itemize}
If the agents are in a dispersed (respectively, general) configuration initially, then we say the agents satisfy a dispersed (respectively, general) initial configuration.  We say that an agent is {\em singleton} if it is alone at a node, otherwise {\em non-singleton} (or multiplicity).

\vspace{2mm}
\noindent{\bf Contributions.}
The leader election results are summarized in Table \ref{table:leaderelection}.
Our results are of two types characterized by whether the elected leader is {\em overtaken} -- the elected leader does not know whether its election is final or if other leaders might overtake it later; nonetheless, the election eventually stabilizes to a unique leader.   

\sloppy

\begin{itemize}
    \item {\bf Stabilizing.} The elected leader may be overtaken. 
    \item {\bf Explicit.} The elected leader is not overtaken. 
\end{itemize}

For the case of $k<n$, we develop two algorithms one stabilizing and another explicit. We design first the stabilizing algorithm and then modify it to make explicit. The stabilizing algorithm works as follows.
For the case of $k<n$, only $k$ nodes of $G$ will be occupied even in a dispersed configuration. Let $C_i\subset G$ be a component (subgraph) of $G$ with exactly one agent positioned per node.  For any two components $C_i,C_j$, if $v_i\in C_i$ then any node $v_j\in C_j, C_j\neq C_i,$ is at least 2 hops away, i.e., $hop(v_i,v_j)\geq 2$. In other words, if $hop(v_i,v_j)=1$, then both $v_i,v_j$ belong to the same component.  For any $k<n$, there will be $1\leq \kappa\leq k$ components $C_1,C_2,\ldots,C_{\kappa}$ such that (i) $C_1\cup C_2\cup \ldots \cup C_{\kappa}=k$ and $C_1\cap C_2\cap \ldots \cap C_{\kappa}=\emptyset$. 
Our leader election algorithms elect
$\kappa$ agents as leaders of $\kappa$ components, one per component.  For time and memory complexities, we need the following parameters which we define now.
\begin{itemize}
    \item Let $|C_i|$ be the number of nodes in $C_i$ (i.e., size of $C_i$) 
    \item Let $\ell_i$ be the number of non-singleton nodes in $C_i$
    \item Let $\zeta_i$ be the number of agents that become local leaders (we define local leaders later) in $C_i$ 
\item    
Let $|C_{max}|:=\max_i |C_i|$, $\ell_{\max}:=\max_i \ell_i$, and $\zeta_{max}:=\max_i \zeta_i$ 
\end{itemize}
The time complexity is $O((|C_{max}|+\log^2k)\Delta)$ 
and memory complexity is $O(\max\{\ell_{\max},\zeta_{\max},\log (k+\Delta)\} \log (k+\Delta))$ bits per agent. 
Therefore, our stabilizing algorithm guarantees that if an elected leader is not overtaken by another leader until $O((|C_{max}|+\log^2k)\Delta)$, then that elected leader remains as a leader, i.e.,  the system as a whole stabilizes to a fixed leader after $O((|C_{max}|+\log^2k)\Delta)$ rounds. Notice that the algorithm does not have parameter knowledge.


\begin{table*}[!t]
\centering
{\footnotesize
\begin{tabular}{|c|c|c|c|}
\toprule
{\bf Algorithm} & {\bf Time} & {\bf Memory/Agent} & {\bf Known} \\ 
\toprule
Stabilizing, $k<n$ &$O((|C_{max}|+\log^2k)\Delta)$ & $O(\max\{\ell_{max},\zeta_{max},\log(k+\Delta)\} \log (k+\Delta))$ & $\times$\\
Explicit, $k<n$ &$O(k\Delta)$&$O(\max\{\ell_{max},\zeta_{max}\} \log (k+\Delta))$ & $k,\Delta$ \\
Explicit, $k=n$ &$O(m)$&$O(\max\{\ell,\zeta, \log n\} \log n)$ & $\times$ \\
\bottomrule
\end{tabular}
\caption{Summary of the results for leader election in the agentic model for $k\leq n$. $|C_{max}|$ is the size of the largest component, $\ell_{\max}$ and $\zeta_{\max}$, respectively, are the maximum number of non-singleton nodes and maximum number of local leaders among components, and $\ell$ and $\zeta$, respectively, are the number of non-singleton nodes and the number of local leaders in $G$ since there is a single component when $k=n$. $m$ is the number of edges in $G$ and $\Delta$ is the maximum degree of $G$.} 
\label{table:leaderelection}
}
\end{table*}

\begin{table*}[!t]
\centering
{\footnotesize
\begin{tabular}{|c|c|c|c|}
\toprule
{\bf Algorithm} & {\bf Time} & {\bf Memory/Agent} & {\bf Known} \\ 
\toprule
Stabilizing, $k<n$ &$O(|C_{max}|(\Delta+\log |C_{max}|)$ & $O(\max\{\ell_{max},\zeta_{max},\Delta,$ & $\times$\\
& $+\Delta \log^2k)$ & $\log(k+\Delta)\} \log (k+\Delta))$& \\
Explicit, $k<n$ &$O(k\Delta)$&$O(\max\{\ell_{max},\zeta_{max},\Delta\} \log (k+\Delta))$ & $k,\Delta$ \\
Explicit, $k=n$ &$O(m+n\log n)$&$O(\max\{\ell,\zeta,\Delta, \log n\} \log n)$ & $\times$ \\
\bottomrule
\end{tabular}
\caption{Summary of the results for MST  in the agentic model for $k\leq n$.} 
\label{table:MST}
}
\end{table*}

We now discuss how we make the above 
stabilizing algorithm for $k<n$ an explicit algorithm.
For this we assume the agents to have parameter knowledge (specifically,  $k$ and $\Delta$).
Given the knowledge of $k$ and $\Delta$, in this algorithm, an agent eligible to become a unique leader in a component $C_i$ waits until round $c_1 \cdot k\Delta$, for some constant $c_1$, before elevating itself as a leader. We will show that if an eligible agent to become a unique leader does not become aware of another eligible agent until $c_1 \cdot k\Delta$, then it can safely become a leader. 
We then prove that the runtime becomes $O(k\Delta)$ and memory becomes $O(\max\{\ell_{max},\zeta_{max}\} \log (k+\Delta))$ bits per agent. 

For the case of $k=n$, we develop an explicit deterministic algorithm for leader election
without requiring any knowledge (neither exact nor upper bound) on parameters $k$, $n$, $\Delta$, and $D$. Notice that for $k<n$, the explicit algorithm for leader elected needed knowledge of $k,\Delta$. Additionally, we elected a single fixed leader (with no overtaking), i.e., $\kappa=1$ (a single component $C$). The time complexity of this algorithm is $O(m)$ rounds and the memory complexity is $O(\max\{\ell,\zeta,\log n\} \log n)$ bits per agent, where $\ell,\zeta$ are the non-singleton nodes in the initial configuration and the number of local leaders in  the single component $C$. 

Using these leader election results for both $k<n$ and $k=n$, we develop stabilizing and explicit algorithms for minimum spanning tree (MST) 
for both $k<n$ and $k=n$. The algorithms are either stabilizing or explicit based on the stabilizing or explicit leader election result used.  The MST results are summarized in Table \ref{table:MST}.

To the best of our knowledge, the results for both leader election and MST for  $k<n$ are established for the first time in this paper. previously, both leader election and MST were studied in the agentic model only for $k=n$. Since our results are for $k\leq n$, our results subsume those results. 

\vspace{2mm}
\noindent{\bf 
Challenges in the Agentic Model.}
In the message-passing model,  in a single round, a node can send a message to all its neighbors and receive messages from all its neighbors. 
Consider the problem of leader election. In the message-passing model, since nodes have IDs, in $D-1$ rounds, all nodes can know the IDs of all other nodes, where $D$ is the diameter of the graph. They then can simply pick the smallest/highest ID node as a leader, solving leader election in $O(D)$ rounds \cite{Peleg90L}. 
In contrast, in the agentic model, the messages from an agent, if any, that are to be sent to the agents in the neighboring nodes have to be delivered by the agent visiting those neighbors. Furthermore, it might be the case that when the agent reaches that node, the agent at that node may have already moved to another node.  
Therefore, it is not clear how a leader could be elected, and if possible to do so, how much would be the time complexity. Additionally, in the agentic model, we have that $k<n$ or $k=n$. For the $k<n$ case, even in a dispersed configuration, there may not be an agent positioned at each node, and the challenge is how to deal with such scenarios.  Therefore, computing in the agentic model is challenging compared to the message-passing model. 

\vspace{2mm}
\noindent{\bf Simulating Message-passing Model to the Agentic Model.}
One may suggest solving tasks in the agentic model simulating the rich set of techniques available in the message-passing model. This is indeed possible when $k=n$ and $n,\Delta$ are known to agents a priori.  When $k=n$ and $n,\Delta$ are known, if agents are not initially in a dispersed configuration, they can be dispersed in $O(n)$ rounds \cite{Kshemkalyani2025}. Since $n$ is known, agents wait until $O(n)$ rounds to start solving graph level tasks.  We can then show that an agent at a node meets all its neighbors in $O(\Delta\log n)$ rounds, knowing both $n,\Delta$. The meeting helps us to guarantee that an agent can communicate with all its neighbors in $O(\Delta\log n)$ rounds. The argument is as follows. Each agent ID is $c\log n$ bits, for some constant $c \geq 1$. Encode the agent ID in $c\log n$ bits. Starting from the most significant bit until reaching the least significant bit, if the current bit is 1, $r_u$ visits all its neighbors in the increasing order of port numbers, which finishes in $2\delta_u$ rounds. If the current bit is $0$, $r_u$ waits at $u$ for $2\Delta$ rounds. Since agent IDs are unique, there must be a round at which when $r_u$ visits a neighbor $v$, the agent $r_v$ is waiting at $v$, hence a {\em meeting}. Therefore, $r_u$ visits all its neighbors in $c\log n  \cdot \Delta=O(\Delta\log n)$ rounds. This meeting means that a round in the message-passing model for $r_u$ can be simulated in $O(\Delta \log n)$ rounds in the agentic model. Therefore, any deterministic algorithm $\mathcal{A}$ that runs for $O(T)$ rounds in the message-passing model can be simulated in the agentic model in $O(T\Delta\log n)$ rounds. Adding the time for dispersion, we have a total time $O(T\Delta\log n+n)$ rounds for simulation in the agentic model.  

Although this simulation seems to work nicely for the agentic model, it suffers from two major problems: (i) assumption of $k=n$ and (ii) assumption of  $n,\Delta$ known to agents a priori. In the agentic model, it may be the case that $k<n$ and $n,\Delta$ may not be known to agents a priori. Actually, it is the quest to design solutions in the agentic model that are oblivious to parameter knowledge. Not knowing the parameters, it is not clear how to meet neighbors. 
Additionally, we would like to solve tasks when $k\leq n$.
It is not known whether complexity bounds better than those through simulation could be obtained by designing algorithms directly in the agentic model. This paper sheds light in this direction.

\vspace{2mm}
\noindent{\bf Techniques.}
We develop a 2-stage technique  
which can solve leader election in the agentic model, even when the parameters are not known and for any $k\leq n$. In Stage 1, the agents compete to become a `local' leader.  
The singleton agents run a {\em Singleton\_Election} procedure to become a local leader. The non-singleton agents run a {\em Non\_Singleton\_Election} procedure to become a local leader. The {\em Singleton\_Election} procedure elects an agent at a node $u$ as a local leader based on 1-hop neighborhood information of $u$. The {\em Non\_Singleton\_Election} procedure first disperses all the co-located agents to different nodes (one per node) running a DFS traversal and finally, an agent becomes a local leader.   It is guaranteed that starting from any initial configuration, at least one agent becomes a local leader in each component $C_i$. 

In Stage 2, the local leaders in each component $C_i$ compete to become a `global' leader at $C_i$.
The local leader runs a {\em Global\_Election} procedure to becomes a global leader. The {\em Global\_Election} procedure runs a DFS traversal and checks for the conditions on whether a local leader can elevate itself as a global leader. The challenge is to handle the possibility of multiple local leaders being elected at different times. Care should be taken so that the whole process does not run into a deadlock situation. We do so by giving priority to the agent that becomes the local leader later in time so that such situations are avoided. In the case of known parameters, the agents run Stage 1 for $O(k\Delta)$ rounds and then Stage 2. In case of unknown parameters, the agents run Stage 2 as soon as Stage 1 finishes.

Denote by  {\em home nodes} the graph nodes where agents became local leaders in Stage 1. $Global\_Election$ for a local leader starts from its home node and ends at its home node. During {\em Global\_Election}, the home nodes of the local leaders become empty since the agents are traversing $G$. Therefore, when a {\em Global\_Election} procedure for a local leader reaches an empty node, it needs to confirm whether the empty node is, in fact, {\em unoccupied} (no agent was ever settled at that node) or {\em occupied} (a home node of some local leader that is currently away from home running {\em Global\_Election}).
This situation also applies for   a {\em Non\_Singleton\_Election} procedure when it reaches an empty node; it needs to confirm whether the node is unoccupied or occupied.  
We overcome this difficulty by asking the local leaders to keep their home node information at an agent positioned at a neighbor. 

However, this technique needs to handle two situations. 
There may be the case the there is no neighboring agent. There may also be the case the one neighboring agent may be responsible for storing the home node information about multiple local leaders.
We develop an approach that allows the local leader in the first situation to become leader without running $Global\_Election$. It becomes $non\_candidate$ if some other agent running $Global\_Election$ or $Non\_Singleton\_Election$ visits it. We handle the second situation through the wait and notify approach.  
Suppose one agent is responsible for storing the home node information for multiple local leaders. We ask it to store the information about only one local leader. All other local leaders wait. After a while that node becomes free (i.e., it does not need to store anymore the home node information of the current local leader since that leader finished $Global\_Election$). It then picks the local leader among the ones waiting to  start $Global\_Election$ and keeps its home node information. All other local leaders in the group become $non\_candidate$. The node that keeps home node information {\em oscillates} on the edge connecting it to the home node of the local leader running $Global\_Election$. We will show that the wait time is bounded for each local leader and hence our claimed bounds are achieved.  


After electing a leader in each component $C_i$, as an application, we use it to solve another fundamental problem of MST construction. 

\vspace{2mm} 
\noindent{\bf 
Related Work.}
%
%
In the message-passing model, the leader election problem was first stated by Le Lann \cite{Lann77} in the context of token ring networks, and since then it has been central to the theory of distributed computing and studied heavily in the literature.  Gallager, Humblet, and Spira \cite{Gallager83} provided a deterministic algorithm for any $n$-node graph $G$ with time complexity $O(n\log n)$ rounds and message complexity $O(m + n\log n)$. Awerbuch \cite{Awerbuch87} provided a deterministic algorithm with time complexity $O(n)$ and message complexity $O(m+n\log n)$. Peleg \cite{Peleg90L} provided a deterministic algorithm with optimal time complexity $O(D)$ and message complexity $O(mD)$. Recently, an algorithm was given in \cite{KPP0T15} with message complexity $O(m)$ but no bound on time complexity, and another algorithm with $O(D\log n)$ time complexity and $O(m\log n)$ message complexity. Additionally, it was shown in~\cite{KPP0T15} that the message complexity lower bound of $\Omega(m)$ and time complexity lower bound of $\Omega(D)$ for deterministic leader election in graphs.  

In the message-passing model, the minimum spanning tree (MST) problem was first studied by  Gallager, Humblet, and Spira  \cite{Gallager83}. They gave  
a deterministic algorithm with  
time complexity $O(n\log n)$ and message complexity $O(m + n\log n)$. The time complexity was improved to $O(n)$ in \cite{Awerbuch87}  and to 
$O(\sqrt{n}\log ^{*}n+D)$ in \cite{GarayKP93,KuttenP98}. Peleg and Rubinovich \cite{PelegR00} established a time complexity lower bound of  
$\Omega (\sqrt {n}/{\log n}+D)$. 


In the agentic model,  gathering is the most studied problem which asks for agents initially positioned arbitrarily on the graph to be positioned at a single node not known a priori.  %
%
The recent results are  \cite{MMM23,Ta-ShmaZ14} that solve Gathering under known $n$. 
Ta-Shma and Zwick \cite{Ta-ShmaZ14} provided a $\tilde{O}(n^5\log \beta)$ time solution to gather $k\leq n$ agents in arbitrary graphs, where $\tilde{O}$ hides polylog factors and $\beta$ is the smallest ID among agents. 
Molla {\it et al.} \cite{MMM23} provided improved time bounds for large values of $k$ assuming $n$ is known but not $k$: (i) $O(n^3)$ rounds, if $k \geq \floor*{\frac{n}{2}} +1$ (ii) $\Tilde{O}(n^4)$ rounds, if $ \floor*{\frac{n}{2}} +1 \leq k < \floor*{\frac{n}{3}} +1$, and (iii) $\Tilde{O}(n^5)$ rounds, if $ \floor*{\frac{n}{3}} +1 > k$. Each agent requires $O(M + m \log n)$ bits, where $M$ is the memory required to implement the universal traversal sequence (UXS) \cite{Ta-ShmaZ14}. 

The opposite of Gathering is the problem of Dispersion which asks the $k\leq n$ agents starting from arbitrary initial configurations in the agentic model to be positioned on $k$ distinct nodes (one per node). Notice that we also solve Dispersion in this paper while electing a leader. In fact, for the case of $k=n$, solving Dispersion while electing a leader helps to make a single component $C$ of $k=n$ nodes and hence only a single unique leader could be elected. For the case of $k<n$, Dispersion helps to come up with the minimum number of components possible, for the given initial configuration.  Dispersion has been introduced by Augustine and Moses Jr. \cite{Augustine:2018}.  The state-of-the-art is a $O(k)$-round solution with $O(\log(k+\Delta))$ bits at each agent without any parameter knowledge, see \cite{Kshemkalyani2025}.  

Maximal Independent Set (MST) and Maximal Dominating Set (MDS) problems were also studied recently in the agentic model \cite{PattanayakBCM24,ChandMS23} assuming $k=n$ and parameters $n,\Delta$ (additionally $m$ and $\gamma$, the number of clusters in the initial configuration, for MDS \cite{ChandMS23}) are known to agents a priori. 
Leader election and MST were studied in the agentic model for the case of $k=n$ with known parameters in \cite{KshemkalyaniICDCIT25} and without known parameters in \cite{kshemkalyani2024agent,KshemkalyaniAAMAS25}. The proposed results in this paper are the first for any graph level task, including leader election and MST, in the agentic model for the case of $k<n$. 

\vspace{2mm}
\noindent{\bf Paper Organization.}
We discuss the stabilizing and explicit deterministic algorithms for leader election for $k<n$ in Section \ref{section:leaderkn}. An explicit deterministic algorithm for leader election for $k=n$ is discussed in Section \ref{section:leader}. We then discuss stabilizing and explicit algorithms for MST for $k\leq n$ in Section \ref{sec: MST construction}. Finally, we conclude in Section \ref{section:conclusion} with a short discussion.

\begin{algorithm}[bt!]
{
\footnotesize
\SetKwInput{KwInput}{Input}
\SetKwInput{KwStates}{States}
\SetKwInput{KwEnsure}{Ensure}
\KwInput{A set $\cR$ of $k< n$ agents with unique IDs positioned initially arbitrarily on the nodes of an $n$-node, $m$-edge anonymous graph $G$ of degree $\Delta$.}

\KwEnsure{Let $C$ be a component (subgraph) of $G$ such that exactly one agent is positioned on each node of $C$. An agent in each connected component $C$ of $G$ of agents  is elected as a leader with status $leader$.}          

\KwStates{Initially, each agent $\psi(u)$ positioned at node $u$ has $\psi(u).status\leftarrow candidate$, $\psi(u).init\_alone \leftarrow true$ if alone at $u$,  $\psi(u).init\_alone \leftarrow false$ otherwise, and $\psi(u).all\_component\_edges\_visited\leftarrow false$. The $init\_alone$ variable is never updated for $\psi(u)$ throughout the algorithm, but the $status$ variable can take values $\in\{candidate, non\_candidate, local\_leader, leader\}$ and the $all\_component\_edges\_visited$ variable can take the value $true$.  
} 












\If{$\psi(u).status=candidate$}
{   

    
    \If{$\psi(u).init\_alone=true$}
    { 
        $Singleton\_Election(\psi(u))$ (Algorithm \ref{algorithm:local_electionk})
    }
    \Else
    {
    \If{$\psi(u)$ is the minimum ID agent at $u$}
    {
        $Non\_Singleton\_Election(\psi(u))$ (Algorithm \ref{algorithm:dispersionk})
    }
}
}
\If{$\psi(u).status=local\_leader$}
{
$Oscillation(\psi(u))$ (Algorithm \ref{algorithm:oscillation})\\
$Global\_Election(\psi(u))$ (Algorithm \ref{algorithm:leader-election})
}
}
\caption{Leader election for agent $\psi(u)$ positioned at node $u$}
\label{algorithm:leader_election}
\end{algorithm}

\begin{algorithm}[bt!]
{
\footnotesize
$N(u) \leftarrow$ neighbors of node $u$, $|N(u)|=\delta_u$\\ 


    $\psi(u)$ visits all neighbors in $N(u)$ (in the increasing order of port numbers)\\ 
    \If{$\forall v\in N(u)$, $\delta_u<\delta_v$}
    {
    \If {no node in $N(u)$ is occupied or all occupied nodes in $N(u)$ are singleton}
    {
    $\psi(u).status\leftarrow local\_leader$
    }
    \Else
    {
    $\psi(u).status\leftarrow non\_candidate$
    }
    }
    \If{$\forall v\in N(u)$, $\delta_u$ is the smallest but $\exists v \in N(u)$ with $\delta_v=\delta_u$}
    {
    \If {(at least one such) $v$ is empty}
    {
    $Neighbor\_Exploration\_with\_Padding(\psi(u))$ (Algorithm \ref{algorithm:padding})\\
    \If {$v$ is still empty and all occupied nodes in $N(u)$ are singleton
    }
    {
    $\psi(u).status\leftarrow local\_leader$
    }
    }
    \If{no such $v$ is empty}
    {
    \If{all occupied nodes in $N(u)$ are singleton (including $\psi(v)$), and $\psi(u).ID>\psi(v).ID$}
    {
    $\psi(u).status\leftarrow local\_leader$
    }
    \Else
    {
    $\psi(u).status\leftarrow non\_candidate$
    }
    }
    }
}

        \caption{$Singleton\_Election(\psi(u))$}
\label{algorithm:local_electionk}
\end{algorithm}

\begin{algorithm}[bt!]
{
\footnotesize
$b\leftarrow$ number of bits in the ID of $\psi(u)$

$b+2b^2 \leftarrow$ number of bits in the ID of $\psi(u)$ after padding a sequence of `10' bits $b^2$ times after the LSB in the original $b$-bit ID

starting from MSB and ending on LSB, if the bit is `1', $\psi(u)$ visits the nodes in $N(u)$ in the increasing order of port numbers, otherwise $\psi(u)$ stays at $u$ for $2\delta_u$ rounds 



        }
        \caption{$Neighbor\_Exploration\_with\_Padding(\psi(u))$}
 \label{algorithm:padding}
\end{algorithm}

\begin{algorithm}[t!]
{
\footnotesize
        run $DFS(r_{min})$ until all $\mathcal{R}(u)$ agents initially at $u$ settle at $|\mathcal{R}(u)|$ different unoccupied nodes of $G$ (one per node). The minimum ID agent $r_{min}$  which settles last  becomes $local\_leader$. All others become $non\_candidate$ once settled\\
        \If{$DFS(r_{min})$ reaches an empty node} 
        {run $Confirm\_Empty()$ procedure (Algorithm \ref{algorithm:confirm-empty}) to classify that node as 
        occupied or unoccupied}  
        } 
\caption{$Non\_Singleton\_Election(r_{min})$}
         \label{algorithm:dispersionk}
\end{algorithm}

\begin{algorithm}[bt!]
{
\footnotesize
        $x\leftarrow$ empty node  $DFS(\psi(u))$ is currently visiting\\ 
        wait at $x$ for a round\\
        \If{$\psi(u)$ does not meet any agent} 
        {$x$ is unoccupied}
        \Else
        {x is occupied}
}
        \caption{$Confirm\_Empty(\psi(u))$}
 \label{algorithm:confirm-empty}
\end{algorithm}

\begin{algorithm}[bt!]
{
\footnotesize
\If{$\psi(u)$ became local leader through $Singleton\_Election(\psi(u))$}
{
\If{$\psi(u)$ found at least a node $w\in N(u)$ with a settled agent $\psi(w)$ while running $Singleton\_Election(\psi(u))$}
{
    \If{$\psi(w)$ is currently not oscillating}
        {
            $\psi(w)$ keeps the information that node $u$ is  $home(\psi(u))$ \\
            $\psi(u)$ starts $Global\_Election(\psi(u))$ (Algorithm \ref{algorithm:leader-election}) and asks 
            $\psi(w)$ to oscillate on the edge $(w,u)$\\
        }
    \Else
        {
            $\psi(u)$ waits until $\psi(w)$ finishes oscillation and notifies $\psi(u)$\\
            \If{multiple local leaders waiting for $\psi(w)$ to finish oscillation}
            {
            \If{$\psi(u)$ has highest priority among the waiting local leaders}
            {
            $\psi(u)$ starts $Global\_Election(\psi(v))$ and asks
            $\psi(w)$ to oscillate on the edge $(w,u)$
            }
            \Else{
            $\psi(u)$ become $non\_candidate$
            }  
            }
        }
 }

  
\Else
 {
 $\psi(u).status\leftarrow leader$\\
 \If{$\psi(u)$ meets any agent running $Global\_Election()$ or $Non\_Singleton\_Election()$}
  {$\psi(u).status\leftarrow non\_candidate$}
  }

}
\If{$\psi(u)$ became local leader through $Non\_Singleton\_Election(\psi(u))$}
{
$w\leftarrow$ the parent node of $u$ in DFS tree built by $DFS(\psi(u))$\\
   \If{$w$ has an agent $\psi(w)$ which is a $non\_candidate$}
        {
            $\psi(u)$ starts $Global\_Election(\psi(u))$ and asks $\psi(w)$ to oscillate on the edge $(w,u)$
        }

    \If{$w$ has an agent $\psi(w)$ which is a $local\_leader$}
        {
            $\psi(u)$ waits until $\psi(w)$ notifies $\psi(u)$ that it became a $non\_candidate$
        }

  \If{$w$ is empty and $\psi(w')$ is oscillating to $w$}
        {
            $\psi(u)$ waits until $\psi(w')$ notifies $\psi(u)$ that it finished oscillating (i.e., $\psi(w)$ returned to $w$ finishing $Global\_Election$)\\
            \If{no local leader agent $\psi(w'')$ is asking $\psi(u)$ to notify}
            {
            $\psi(u)$ starts $Global\_Election(\psi(u))$ and asks $\psi(w)$ to oscillate on the edge $(w,u)$
            }
            \Else
            {$\psi(u)$ notifies $\psi(w'')$ and becomes $non\_candidate$ }
        }
    
 }
}
\caption{$Oscillation(\psi(u))$}
\label{algorithm:oscillation}
\end{algorithm}

\begin{algorithm}[bt!]
 {
 \footnotesize
        run $DFS(roundNo_u,\psi(u))$, where $roundNo_u$ is the round $\psi(u)$ became local leader, and try to explore component $C(\psi(u))$. \\

\If{$DFS(roundNo_u,\psi(u))$ head reaches an empty node} 
{ 
 
\If{$Confirm\_Empty()$ returns occupied}
{continue $DFS(roundNo_u,\psi(u))$ in the forward phase}
\Else
{backtrack to the node from which $DFS(roundNo_u,\psi(u))$ entered that node}
}        
        \If{$\psi(u).all\_component\_edges\_visited=true$}
        {
        $\psi(u).status\leftarrow leader$; return $home(\psi(u))$\\
        }
        \ElseIf {head of $DFS(roundNo_u,\psi(u))$ meets    $DFS(roundNo_v,\psi(v))$ (head or nonhead)}
        {
        \If{$roundNo_u>roundNo_v$ or $(roundNo_u=roundNo_v \And \psi(u).ID>\psi(v).ID)$}
        {
        $DFS(roundNo_u,\psi(u))$ continues
        }
        \Else
        {
        $DFS(roundNo_u,\psi(u))$ stops, $\psi(u)$ becomes $non\_candidate$ and  returns $home(\psi(u))$
        }}
        \ElseIf{$DFS(roundNo_u,\psi(u))$ meets  head of $DFS(\psi(v))$ running  $Non\_Singleton\_Election(\psi(v))$ (Algorithm \ref{algorithm:dispersionk})}
        {
        $DFS(roundNo_u,\psi(u))$ stops, $\psi(u)$ becomes $non\_candidate$ and returns $home(\psi(u))$
        }
        
}
   \caption{$Global\_Election(\psi(u))$}
    \label{algorithm:leader-election}       
\end{algorithm}

\begin{figure}[bt!]
    \centering
    \includegraphics[width=1.0\linewidth]{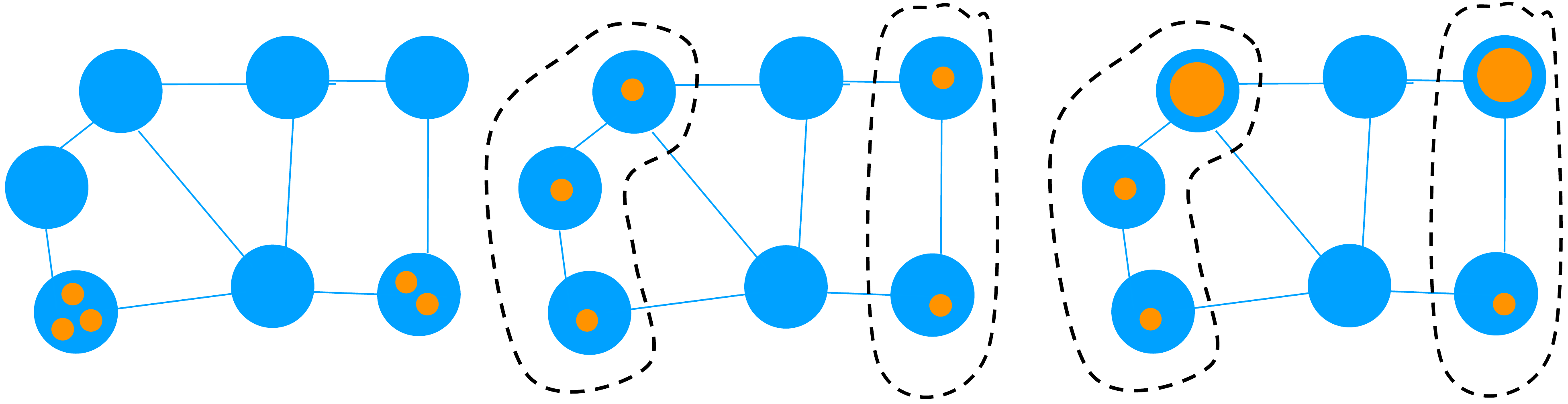}
    \caption{An illustration of leader election. ({\bf left}) A general initial configuration of 5 agents on a 7-node graph $G$ with 2 non-singleton nodes and no singleton node.  ({\bf middle}) Two components $C_1$ and $C_2$ are formed (shown with dashed lines) after two $Non\_Singleton\_Election$ procedures run by two non-singleton nodes finish; agent are now in a dispersed configuration in each component. ({\bf right}) An agent at each component becomes a global leader (shown as big circles inside blue circles). Each agent in $C_1$ is  (at least) 2 hops apart from an agent in $C_2$ (and vice-versa).}
    \label{fig:components}
\end{figure}

\section{Leader Election, $k<n$}
\label{section:leaderkn}
In the message-passing model, the case of $k<n$ does not occur, i.e., each node corresponds to a static computational device. In contrast, in the agentic model, when $k<n$, a node may not necessarily correspond to a computational device even in a dispersed configuration. 
%
We present our leader election algorithm which, starting from any initial configuration (dispersed or general) of $k<n$ agents on an $n$-node graph $G$, ensures that one agent is elected as a global leader in each component $C_i$, where $C_i$ represents a subgraph of $G$ on which each node has an agent positioned. 
As a byproduct, our algorithm transforms the general initial configuration to a dispersed configuration.
%
%
%
Notice that there will be $1\leq \kappa\leq k$ leaders elected if $\kappa$ components are formed in the dispersed configuration. 
Fig.~\ref{fig:components} provides an illustration of leader election for 5 agents initially positioned at two non-singleton nodes of a 7-node graph $G$. 
We develop a stabilizing algorithm (which does not need any parameter knowledge, but overtaking of the  elected leader in each component may occur for a certain time before stabilizing to a single leader). We will discuss later how to make it explicit (using knowledge of $k,\Delta$ but avoiding overtaking of the elected leader in each component).  


\subsection{Stabilizing Algorithm}
The pseudocode of the algorithm is given in Algorithm \ref{algorithm:leader_election}. 
Initially, a graph node may have zero, one, or multiple agents. An agent at each node is `candidate' to become a leader. 
%
Our algorithm runs in 2 stages, Stage 1 and Stage 2. Stage 2 runs after Stage 1 finishes. 
In Stage 1, if an agent is initially singleton at a node, it runs {\em Singleton\_Election} to become a `local leader'. However, if an agent is initially non-singleton, then it runs {\em Non\_Singleton\_Election} to become a local leader. As soon as becoming a local leader, Stage 2 starts in which the local leader agent runs {\em Global\_Election} to become a `global leader'. To do so, a local leader may need to wait for a while. After the bounded wait, it either directly becomes non\_candidate or starts to run $Global\_Election$.   
After {\em Global\_Election} finishes, a local elevates itself to a global leader and returns to its home node (if it is not at home already). 
The elected global leader may be overtaken by another global leader but we show that the process stabilizes to a global leader after $O((|C_{max}|+\log^2k)\Delta)$ rounds in each component $C_i$, where  $|C_{max}|:=\max_i |C_i|$.
The memory needed in $O(\max\{\ell_{max},\zeta_{\max},\log (k+\Delta)\} \log (k+\Delta))$ bits per agent. Here 
$\ell_{\max}:=\max_i \ell_i$ and $\zeta_{max}:=\max_i \zeta_i$, with $\ell_i,\zeta_i$, respectively, being the  number of non-singleton nodes in $C_i$ and the number of agents that become local leaders  in $C_i$. 
We discuss Stage 1 and Stage 2 in detail separately below.

\sloppy
\subsubsection{Stage 1 -- Local Leader Election}
Stage 1 concerns with electing local leaders among $k$ agents. Stage 1 differs for singleton and non-singleton agents.  Singleton agents run {\em Singleton\_Election} and non-singleton agents run {\em Non\_Singleton\_Election}. 
We describe them separately below.

\vspace{2mm}
\noindent{\bf Singleton Election.}
The pseudocode of {\em Singleton\_Election} is in Algorithm \ref{algorithm:local_electionk}. Let $r_u$  be the agent positioned at node $u$; we write $\psi(u)=r_u$. 
For $\psi(u)$ to eligible to become a local leader, 
either all neighbors in $N(u)$ have to be empty or all the non-empty neighbors in $N(u)$ must have a singleton agent. If this condition is true, then $\psi(u)$ becomes a local leader if 
one of the following conditions is satisfied:
\begin{itemize}
\item {\bf Condition I.} $\forall v\in N(u), \delta_v>\delta_u$. In other words, $u$ has the unique smallest degree among the nodes in $N(u)$. 
\item {\bf Condition II.} $\forall v\in N(v), \delta_v\geq \delta_u$ with no node $v'\in N(u)$ with $\delta_{v'}=\delta_u$ is empty, and $\psi(u).ID>\psi(v').ID$. In other words, $\psi(u)$ has a higher ID than $\psi(v')$ when $u$ and $v'$ have the same smallest degree.
\item {\bf Condition III.} $\forall v\in N(v), \delta_v\geq \delta_u$ with (at least one) node $v'\in N(u)$ with $\delta_{v'}=\delta_u$ is empty, $\psi(u)$ runs $Neighbor\_Exploration\_with\_Padding(\psi(u))$ (Algorithm \ref{algorithm:padding}) and $v'$ remains empty even after  Algorithm \ref{algorithm:padding} finishes. In other words, $\psi(u)$ finds that the same smallest degree neighbor $v'$ is empty. 
\end{itemize}


If none of the above conditions are satisfied, $\psi(u)$ becomes a `non-candidate'.  
$\psi(u)$ checks for {\bf Conditions I} and {\bf II} after visiting all its neighbors in $N(u)$ starting from port 1 and ending at port $\delta_u$. Notice that visiting a neighbor finishes in 2 rounds, one round to reach to the neighbor and one round to return. If {\bf Condition II} does not satisfy because of empty $v'\in N(u)$ (same smallest degree as $u$), $\psi(u)$ checks for {\bf Condition III} after running $Neighbor\_Exploration\_with\_Padding(\psi(u))$ (Algorithm \ref{algorithm:padding}). We now discuss how $Neighbor\_Exploration\_with\_Padding(\psi(u))$ executes.

\vspace{2mm}
\noindent{\bf Same Smallest Degree Neighbor Exploration.}
When $\psi(u)$ shares its smallest degree among $N(u)$ with at least a node $v\in N(u)$, i.e., $\delta_u=\delta_{v}$, {\bf Condition I} cannot be used. We need to decide where {\bf Condition II} or {\bf Condition III} can be used. $\psi(u)$ needs to determine whether $v$ is (i) empty (ii) singleton or (iii) non-singleton.  To confirm this, if there is $\psi(v)$, then $\psi(u)$ needs to meet $\psi(v)$ or vice versa.
Knowing $k,\Delta$, this meeting can be done in $O(\Delta \log k)$ rounds as discussed under the message-passing model simulation in Section \ref{section:introduction}.
However, not knowing the parameters, this meeting is challenging and we develop the following meeting technique (Algorithm \ref{algorithm:padding}).
Suppose agent $\psi(u)$ has an ID of $b$ bits; note that $b\leq c\cdot \log k$ for some constant $c>1$. We pad a sequence of `10' bits $b^2$ times  
after the LSB of the $\psi(u)$'s ID, i.e., 
\begin{eqnarray*}
b\underbrace{10}_{\mbox{$1$}}\underbrace{10}_{\mbox{$2$}}\ldots\underbrace{10}_{\mbox{$b^2$}}.
\end{eqnarray*}

Now the ID of $b$ bits becomes the ID of $b+2b^2$ bits. 
Agent $\psi(u)$ starts from its MSB bit and ends at LSB bit. If the current bit is 1, then $\psi(u)$ visits all neighbors in $N(u)$ in the increasing order of port numbers, which finishes in $2\delta_u$ rounds. However, when the current bit is $0$, it remains at $u$ for $2\delta_u$ rounds. (Notice that $\phi(v)$, if any, also runs the same procedure padding appropriately its ID bits.) 
We will show that using this padding approach $\psi(u)$ meets $\psi(v)$ or vice-versa in $O(\Delta\log^2n)$ rounds, a crucial component in Algorithm \ref{algorithm:local_electionk} ($Singleton\_Election$). After Algorithm \ref{algorithm:padding} finishes, then $\psi(u)$ either becomes a local leader or a non-candidate applying either {\bf Condition II} or {\bf Condition III}.

\vspace{2mm}
\noindent{\bf Non-Singleton Election.}
The pseudocode of
{\em Non\_Singleton\_Election} is in Algorithm \ref{algorithm:dispersionk}. Suppose a node $v$ has $k'=|\mathcal{R}(v)|$ agents initially. Let $r_{min}$ be the smallest ID agent in $\mathcal{R}(v)$. 
Agent $r_{min}$ will become a local leader after setting agents in $\mathcal{R}(v)\backslash\{r_{min}\}$ to the nodes of $G$, one per node. The largest ID agent in $\mathcal{R}(v)$ stays at $v$ and becomes $\psi(v)$. Agent $r_{min}$ runs a DFS algorithm, we denote as $DFS(r_{min})$. As soon as an unoccupied node is visited, the largest ID agent in the group settles and others continue $DFS(r_{min})$. An empty node visited in confirmed as occupied or unoccupied running Algorithm \ref{algorithm:confirm-empty} ($Confirm\_Empty()$).
As soon as $r_{min}$ reaches to an unoccupied node $u$ alone, it declares itself as a local leader $\psi(u)$.  All the agents in $\mathcal{R}(v)\backslash\{\psi(u)\}$ become $non\_candidate$.  There are two issues to deal with in this process.

\begin{itemize}
\item {\bf Issue 1 -- Identifying an empty node visited occupied or unoccupied:} Suppose $DFS(r_{min})$ reaches an empty node $x$. An empty node can be of three types: (I) unoccupied
(II) home of a singleton agent running $Singleton\_Election$, (III) home of a local leader running $Global\_Election$. 
$Non\_Singleton\_Election$ that runs $DFS(r_{min})$ can only settle an agent on an unoccupied node. This is done through $Confirm\_Empty(r_{min})$ (Algorithm \ref{algorithm:confirm-empty}) procedure. At each empty node $x$ visited, $DFS(r_{min})$ waits at $x$ for a round. 
If an agent at $x$ is doing $Singleton\_Election$, it will return to $x$ within a round. 
If an agent at $x$ was a local leader (either through $Singleton\_Election$ or $Non\_Singleton\_Election$) doing $Global\_Election$, there is an agent at $N(x)$ that is oscillating to $x$ and that will be at $x$ in the next round. 
Therefore, if $x$ is empty even after waiting for one round, then $x$ is unoccupied, and  $DFS(r_{min})$ can settle an agent at $x$. 
\item {\bf Issue 2 -- \bf $DFS(r_1)$ meets $DFS(r_2)$:} If the head of $DFS(r_1)$ meets the head of $DFS(r_2)$ at a node $w$, then, if $r_1.ID>r_2.ID$, then the highest ID agent belonging to $DFS(r_1)$ stays at $w$ (and becomes $non\_candidate$), otherwise the highest ID agent  belonging to $DFS(r_2)$ stays at $w$ (and becomes $non\_candidate$). The winning DFS continues and the losing DFS stops and hands over the unsettled agents to the winning DFS to continue dispersing the agents. 
\item {\bf Issue 3 -- \bf $DFS(r_1)$ meets $DFS(roundNo_v,\psi(v))$:} This case is an example of $Non\_Singleton\_Election$ meeting $Global\_Election$. $DFS(r_1)$ continues.  If the meeting is at the head of $Global\_Election$ DFS, the agent doing   $Global\_Election$ DFS  becomes $non\_candidate$ and returns to its home node.     
    \end{itemize}

\subsubsection{Stage 2 -- Global Leader Election}
Stage 2 concerns electing global leaders among local leaders.   Stage 2 for a local leader agent $\psi(u)$ runs $Oscillation$ and {\em Global\_Election} (the pseudocodes are in Algorithms \ref{algorithm:oscillation} and \ref{algorithm:leader-election}).
The goal of running $Oscillation$ is to see whether a local leader can in fact run $Global\_Election$ or become a $non\_candidate$ even before running $Global\_Election$. For the local leaders that can run $Global\_Election$,  
the goal of running {\em Global\_Election} is to see whether $DFS(roundNo_u,\psi(u))$ can explore $C(\psi(u))$, the component $\psi(u)$ belongs to. Exploration here means visiting all the edges that belong to $C(\psi(u))$. Additionally, $roundNo_u$ denotes the round at which $\psi(u)$ became a local leader which  helps in preventing deadlock situations when $DFS(roundNo_u,\psi(u))$ meets $DFS(roundNo_v,\psi(v))$ from another local leader $\psi(v)$.  Except the use of $roundNo$, $DFS(roundNo_u,\psi(u))$ is the same as $DFS(\psi(u))$  we use in $Non\_Singleton\_Election$.
If $DFS(roundNo_u,\psi(u))$ can explore $C(\psi(u))$, it elevates itself  as a `global' leader and returns to  $home(\psi(u))$ (if it is not already at $home(\psi(u))$).  Agent $\psi(u)$ remains as a global leader if no other local leader meets $\psi(u)$ later in time. If another local leader $\psi(v)$ meets $\psi(u)$, then we say that overtaking occurred for $\psi(u)$. $\psi(u)$ (the current global leader) now becomes $non\_candidate$. We will show that overtaking does not occur after $roundNo=c_1\cdot (|C_{max}|+\log^2k) \Delta$ rounds, i.e.,  the leader election at each component $C_i$ stabilizes to a single global leader in at most $c_1\cdot (|C_i|+\log^2k) \Delta$ rounds. Notice that this approach does not use any parameter knowledge. 

There are several challenges to overcome. The first challenge to overcome is, once elected a local leader, whether that local leader can run $Global\_Election$ and, if so, when. The second challenge to overcome is how to run $Global\_Election$ provided that there may be multiple $Global\_Election$ procedures concurrently running from different local leaders. There may also be concurrent situations of a $Global\_Election$ meeting $Non\_Singleton\_Election$.  We describe our approaches to overcome these challenges separately below. 
\begin{itemize}
\item {\bf Issue 1 -- Can a local leader $\psi(u)$ run $Global\_Election$ and if so when?} We have two sub-cases.  The pseudocode is described in $Oscillation(\psi(u))$ (Algorithm \ref{algorithm:oscillation}).
\begin{itemize}
\item {\bf Issue 1.A -- $\psi(u)$ became  local leader through {\em Singleton\_Election}:} We further have two sub-cases.
\begin{itemize}
\item  {\bf Issue 1.A.i -- $\psi(u)$ find no neighbor in $N(u)$ occupied:} This is a special case. $\psi(u)$ elevates itself as a global leader and stays at $u$. That means, it does not run $Global\_Election$. Later, if it meets at $u$ another agent running $Global\_Election$ or $Non\_Singleton\_Election$, then $\psi(u)$ becomes $non\_candidate$, since the agent running  $Global\_Election$ or $Non\_Singleton\_Election$ will become leader in later time. 
\item {\bf Issue 1.A.ii -- $\psi(u)$ finds at least a neighbor in $N(u)$ occupied:} Let $w$ be that neighbor with agent $\psi(w)$ positioned. 
$\psi(u)$ now checks whether $\psi(w)$ is oscillating. This oscillation helps in the $Confirm\_Empty$ procedure. If so, there must be another local leader, say $\psi(w')$, $w'\neq w\neq u$, already running $Global\_Election$ and $\psi(w)$ is oscillating on the edge $(w,w')$.  
Therefore, if $\psi(w)$ is not oscillating, $\psi(u)$ starts $Global\_Election$ and asks $\psi(w)$ to oscillate between $w$ and $u$. However, if $\psi(w)$ is oscillating, $\psi(u)$ waits asking $\psi(w)$ to notify once $\psi(w)$ stops oscillating. $\psi(w)$ stops oscillating when the local leader $\psi(w')$ running $Global\_Election$ returns to $w'$. As soon as $\psi(w)$ notifies $\psi(u)$, $\psi(u)$ starts $Global\_Election$ and asks $\psi(w)$ to oscillate on the edge $(w,u)$.  

There are two cases that needs further attention:
The first case is of multiple local leaders finding $\psi(w)$ non-oscillating in the same round. In such case, symmetry is broken choosing the lexicographically highest priority local leader (we use $roundNo$ of when the agent became local leader and its ID together to break symmetry). Others become $non\_candidate$. 
The second case is of multiple local leaders waiting for $\psi(w)$ to stop oscillating.  In such case, after $\psi(w)$ stops oscillating, $\psi(u)$ runs $Global\_Election$ if it is the lexicographically highest priority local leader among those waiting, otherwise it becomes $non\_candiadate$.
\end{itemize}

\item {\bf Issue 1.B -- $\psi(u)$ became  local leader through {\em Non\_Singleton\_Election}:}
Consider the DFS tree $DFS(\psi(u))$ built while running $Non\_Singleton\_Election$. 
Let $w$ be the parent of $u$ in $DFS(\psi(u))$.
We have two sub-cases.
\begin{itemize}
\item {\bf Issue 1.B.i -- Node $w$ has an agent $\psi(w)$ which is a $non\_candidate$:}
In this case, $\psi(u)$ starts $Global\_Election$ and asks $\psi(w)$ to oscillate on the edge $(w,u)$. 

\item {\bf Issue 1.B.ii -- Node $w$ has a local leader agent $\psi(w)$:}
In this case, $\psi(u)$ waits until $\psi(w)$ notifies $\psi(u)$ that $\psi(w)$ became a $non\_candidate$. This is the case of local leader chain in which the subsequent agent in the chain became local leader later in time than the previous agent and hence except the last agent in the chain, all others can become $non\_candidate$.   

\item {\bf Issue 1.B.iii -- Node $w$ is empty and an agent $\psi(w')$ is oscillating to $w$:}
There must be the case that $\psi(w)$ was a local leader (i.e., $w$ happened to be $home(\psi(w'))$ of a  local leader $\psi(w)$) currently running $Global\_Election$. $\psi(u)$ asks $\psi(w')$ to notify once it stops oscillating. $\psi(u)$ waits at $u$ until such notification.  $\psi(w')$ stops oscillating once $\psi(w)$ returns to $w$ after finishing $Global\_Election$. $\psi(w')$ then notifies $\psi(u)$. $\psi(u)$ now checks if there is a local leader agent waiting (Issue 1.B.ii). If so, $\psi(u)$ informs that local leader and becomes $non\_candidate$. Otherwise, $\psi(u)$ starts $Global\_Election$ and asks $\psi(w)$ to oscillate on the edge $(w,u)$.  

\end{itemize}

\end{itemize}

\item {\bf Issue 2 -- How to run $Global\_Election$ handling meetings:}
$Global\_Election(\psi(u))$ of a local leader $\psi(u)$ runs $DFS(roundNo_u,\psi(u))$ to explore $C(\psi(u))$. We first discuss how $DFS(roundNo_u,\psi(u))$ runs. If a forward move  at a node $v$ takes $DFS(roundNo_u,\psi(u))$ to a node $w$ with a settled agent $\psi(w)$, we know that  $DFS(roundNo_u,\psi(u))$ is exploring $C(\psi(u))$. 
Suppose 
$w$ is empty; i.e., a forward move from $v$ took $DFS(roundNo_u,\psi(u))$ to an empty node. There are two situations: (i) either $w$ is a home node (which is currently empty) or (ii) it is an unoccupied node. $DFS(roundNo_u,\psi(u))$ runs $Confirm\_Empty()$ at $w$. If $Confirm\_Empty()$ returns $w$ occupied, $DFS(roundNo_u,\psi(u))$ continues making a forward move from $w$. However, if $Confirm\_Empty()$ returns $w$ unoccupied, $DFS(roundNo_u,\psi(u))$ returns to $v$. It is easy to see that this approach makes $DFS(roundNo_u,\psi(u))$ to explore only the nodes and edges of component $C(\psi(u))$.


During the traversal, $DFS(roundNo_u,\psi(u))$ may meet another DFS traversal. Notice that the another DFS traversal may be of $Global\_Election$ (type $DFS(roundNo_x,\psi(x))$) or of $Non\_Singleton\_Election$ (type $DFS(\psi(x))$).  
We describe separately below how we handle them.
\begin{itemize}
\item {\bf Issue 2.A -- $DFS(roundNo_u,\psi(u))$ meets $DFS(roundNo_x,\psi(x))$:} We give priority to the agent that became local leader later in time, i.e., $DFS(roundNo_x,\psi(x))$ continues if $roundNo_u>roundNo_x$, otherwise  $DFS(roundNo_x,\psi(x))$. There may be the case that $roundNo_u=roundNo_x$, which we resolve through agent IDs, i.e., $DFS(roundNo_u,\psi(u))$ continues if $\psi(u).ID>\psi(x).ID$, otherwise  $DFS(roundNo_x,\psi(x))$.  
If $DFS(roundNo_u,\psi(u))$ stops, $\psi(u)$ becomes a non-candidate, and returns $home(\psi(u))$. 

\item {\bf Issue 2.B -- $DFS(roundNo_u,\psi(u))$ meets head of $DFS(r)$ running $Non\_Singleton\_Election(r)$ (Algorithm \ref{algorithm:leader-election}):}  $DFS(roundNo_u,\psi(u))$ stops, becomes non-candidate, and returns $home(\psi(u))$. This is because agent $r$ is eligible to become a local leader since it will finish  $Non\_Singleton\_Election(r)$ later in time. 
\end{itemize}
\end{itemize}

\subsubsection{Analysis of the Algorithm}
We now analyze our stabilizing algorithm (Algorithm \ref{algorithm:leader_election}). We have $k<n$ and agents have no parameter knowledge.  We start with Stage 1 procedure {\em Singleton\_Election} (Algorithm \ref{algorithm:local_electionk}). Recall that {\em Singleton\_Election} is run by singleton agents.

\begin{lemma}
\label{lemma:padding}
 Suppose $\forall v\in N(u), \delta_v\geq \delta_u$ and $\exists v'\in N(u), \delta_{v'}=\delta_u$. Agent $\psi(u)$ meets a singleton agent $\psi(v')$ (if any positioned initially at node $v'$) or vice-versa in $O(\delta_u \log^2 k)$ rounds, running  $Neighbor\_Exploration\_with\_Padding(\psi(u))$ (Algorithm \ref{algorithm:padding}). 
\end{lemma}

\begin{proof}
Each agent has a unique ID of size $\leq c\cdot \log k$, for some constant $c$. Therefore, for any two neighboring agents $\psi(u)$ and $\psi(v')$, the number of bits on their IDs are either equal or unequal. Consider the equal case first. If $\psi(u)$ and $\psi(v')$ have an equal number of bits (say, $b$) then their IDs must differ by at least a bit, i.e., if one has bit `1' at $\beta$-th place from MSB, another must have bit `0' at $\beta$-th place from MSB. Since an agent explores neighbors in $N(u)$ while the bit is `1' and stays at $u$ when the bit is `0', $\psi(u)$ finds $\psi(v')$ within $2 \delta_{u} \cdot b$ rounds. 
%

Now consider the unequal case. Let $\psi(u)$ and $\psi(v')$, respectively, have  $b$ and $d$ bits in their IDs, $b\neq d$. W.l.o.g., suppose $b>d$, i.e., $b = d+c_1$, $d, c_1 \geq 1$. Consider now the padding. 
The $b$ bit ID becomes $b+2b^2$ bit ID and the $d$ bit ID becomes $d+2d^2$ bit ID. 
We have that 
$$b+2b^2=(d+c_1) + 2(d+c_1)^2 = 2d^2+2c_1^2+4 \cdot d \cdot c_1+d+c_1.$$ 
The difference in the number of bits of the IDs of $\psi(u)$ and $\psi(v')$ is $$2c_1^2+4 \cdot d \cdot c_1+c_1.$$ 

Since $d,c_1\geq 1$, $2c_1^2+4 \cdot d \cdot c_1+c_1$ is at least  $7$-bit long, and out of the $7$ bits, at least  $3$ bits are `1'.
    What that means is, if $\psi(v')$ (the smaller ID than $\psi(u)$) stops after $2\delta_{v}(d+2d^2)$ rounds, then there are at least $3$ chances for $\psi(u)$ to meet $\psi(v')$ at  $v'$. Therefore, the round complexity becomes   $O(2\delta_{u} (b+2b^2))=O(\delta_u \log ^2k)$ rounds, since $b\leq c\cdot \log k$. 
\end{proof}

\begin{lemma}\label{lem:rounds_with_same_degree_leader}
     If $\psi(u)$  becomes a local leader running $Singleton\_Election(\psi(u))$, no singleton agent $\psi(v),~v\in N(u),$ becomes local leader running $Singleton\_Election(\psi(v))$.
     $Singleton\_Election(\psi(u))$ finishes in $O(\delta_u \log^2 k)$ rounds.
    \end{lemma}
\begin{proof}
Suppose $\forall v\in N(u), \delta_v>\delta_u$.  Consider a singleton agent $\psi(v)$ at any node $v\in N(u)$. When $\psi(v)$ runs  $Singleton\_Election(\psi(v))$ it finds that $\delta_u<\delta_v$ and $\psi(v)$ becomes a non-candidate. An unoccupied node $v'\in N(u)$ does not run {\em Singleton\_Election}. Therefore, for $\psi(u)$, visiting all neighbors in $N(u)$ finishes in $2\delta_u$ rounds.

Now  suppose $\forall v\in N(u), \delta_v\geq \delta_u$ and $\exists v'\in N(u), \delta_{v'}=\delta_u$.  We have from Lemma \ref{lemma:padding} that $\psi(u)$ either finds $v'$ has no singleton agent or meets singleton agent $\psi(v')$ in  $O(\delta_u \log ^2k)$ rounds running $Neighbor\_Exploration\_with\_Padding(\psi(u))$ (Algorithm \ref{algorithm:padding}). 
If there is singleton agent $\psi(v')$ at $v'$, then $\psi(u)$ becomes a local leader if $\psi(u).ID>\psi(v').ID$, otherwise a non-candidate. If $v$ is unoccupied, $\psi(u)$ becomes a local leader and no agent in $N(u)$ becomes a local leader. The time complexity is to first visit all the neighbors in $N(u)$, which finishes in $2\delta_u$ rounds, and after that time to run  
$Neighbor\_Exploration\_with\_Padding(\psi(u))$ (Algorithm \ref{algorithm:padding}). Therefore, the total time complexity of $Singleton\_Election(\psi(u))$ is    $2\delta_u+O(\delta_u \log ^2k)=O(\delta_u \log ^2k)$ rounds.   
\end{proof}

 \begin{lemma}
     \label{lemma:dispersed-election}
Consider a dispersed initial configuration. At least one singleton agent $\psi(u)$ at node $u$  becomes a local leader running 
$Singleton\_Election(\psi(u))$
(Algorithm \ref{algorithm:local_electionk}).   \end{lemma}

\begin{proof}
    Let $u$ be the smallest degree node in $G$, i.e., $\delta_u=\min_{v\in G}\delta_v$. 
    We have two cases: (I) $\forall v'\in G, v'\neq u, \delta_u<\delta_{v'}$
    (II) $\exists v'\in G, \delta_{v'}=\delta_u$.
    In Case (I), since $u$ is the unique smallest degree node, $\psi(u)$ becomes a local leader  in $O(\delta_u \log^2 n)$ rounds. No singleton agent $\psi(w)$ in $N(u)$ becomes a local leader since $\delta_w>\delta_u$. Now consider Case (II). We have two sub-cases: (II.A) $v'\in N(u)$ (II.B) $v'\notin N(u)$. In Case (II.B), $\psi(u)$ is elected as a local leader as discussed in Case (I). For Case (II.A), we have from Lemma \ref{lem:rounds_with_same_degree_leader} that $\psi(u)$ meets $\psi(v')$ (or vice-versa) in $O(\delta_u\log^2k)$ rounds. After that, either $\psi(u)$ or $\psi(v')$ becomes a non-candidate depending on their IDs. Suppose $\psi(u).ID>\psi(v').ID$, $\psi(v')$ becomes non-candidate, $\psi(u)$ becomes a local leader, and we are done. Otherwise, if $\psi(u).ID<\psi(v').ID$, $\psi(u)$ becomes a non-candidate. Now suppose even with $\psi(u).ID<\psi(v').ID$, $\psi(v')$ becomes a non-candidate. Then, there must be the case that $v'$ has a neighbor $v''\in N(v'')$ with $\delta_{v''}=\delta_{v'}=\delta_u$ and $\psi(v'').ID>\psi(v').ID$. This chain stops at the first node $v*$ such that $\psi(v''').ID< \psi(v*).ID>\ldots>\psi(v'').ID>\psi(v').ID>\psi(u).ID$, $v'''\neq v*\neq \ldots \neq v''\neq v'\neq u$, if $v'''$ has a singleton agent. If $v'''$ is empty, this chain stops at $v*$. In both cases, agent $\psi(v*)$ in this chain becomes a local leader.    
 \end{proof}

 We now analyze Stage 1 procedure $Non\_Singleton\_Election$ (Algorithm \ref{algorithm:dispersionk}). Recall that $Non\_Singleton\_Election$ is run by non-singleton agents. 
 The $Non\_Singleton\_Election$ procedure for an agent $r_{min}$ runs $DFS(r_{min})$.

 \begin{lemma}
\label{lemma:oscillatingneighbor}
 Let $w$ be an occupied node that is currently empty.
 There exists an agent  at $w'\in N(w)$ that visits $w$ every 2 rounds.  
\end{lemma}
\begin{proof}
The occupied node $w$ is empty in the following situations.
\begin{itemize}
\item The agent $\psi(w)$ at $w$ is doing $Singleton\_Election$.
\item The agent $\psi(w)$ at $w$ is a local leader doing $Global\_Election$.
\end{itemize}
In the first case, $\psi(w)$ returns $w$ in the next round. In the second case, we consider two sub-cases: 
\begin{itemize}
    \item {\bf $\psi(w)$ became local leader through $Singleton\_Election$:} Suppose all neighbors in $N(w)$ were empty when $\psi(w)$ became the local leader. In this case, $\psi(w)$ never leaves $w$. If at least one neighbor $w\in N(w)$ is non-empty, Algorithm \ref{algorithm:oscillation} guarantees that the  settled agent $\psi(w')$ oscillates on the edge $(w',w)$, if $\psi(w)$ leaves $w$. 
    \item {\bf $\psi(w)$ became local leader through $Non\_Singleton\_Election$:} Algorithm \ref{algorithm:oscillation} again guarantees that the settled agent $\psi(w')$ oscillates on the edge $(w',w)$, if $\psi(w)$ leaves $w$.
\end{itemize}
\end{proof}

\begin{lemma}
\label{lemma:confirmempty}
 Suppose $DFS(r_{min})$ reaches an empty node $w$.  The $Confirm\_Empty()$ procedure (Algorithm \ref{algorithm:confirm-empty}) confirms in an additional round whether $w$ is unoccupied or occupied. 
\end{lemma}
\begin{proof}
Throughout the algorithm, 
an empty node $w$ may be one of the three types: 
\begin{itemize}
    \item [(I)] an unoccupied node, 
    \item [(II)] home of an agent running $Singleton\_Election$, 
    \item [(III)] home of a local leader agent running $Global\_Election$, 
\end{itemize}
For Case (II), waiting for a round at $w$ is enough since $w$ is the home of the agent running $Singleton\_Election$ and hence the agent comes home every 2 rounds. For Case (III), Algorithm \ref{algorithm:oscillation} guarantees that  there is a neighboring agent at node $w'\in N(w)$ oscillating between $w'$ and $w$ which finishes in every 2 rounds (Lemma \ref{lemma:oscillatingneighbor}). 
\end{proof}
 
\begin{lemma}
\label{lemma:dispersion}
    $Non\_Singleton\_Election(r_{min})$ (Algorithm~\ref{algorithm:dispersionk}) run by an initially non-singleton agent $r_{min}$ of minimum ID among the $x>1$ co-located agents finishes dispersing $x$ agents to $x$ different nodes of $G$ in $O(|C(r_{min})|\Delta)$ rounds with $O(\ell_{C(r_{min})} \log (k+\Delta))$ bits memory per agent, where $\ell_{C(r_{min})}$ is the number of non-singleton nodes in component $C(r_{min})$ agent $r_{min}$  belongs to. 
\end{lemma}
\begin{proof}
Let $\ell$ be the number of non-singleton nodes in the initial configuration. $\ell$ DFSs will be run by $\ell$ minimum ID agents, one each from $\ell$ non-singleton nodes of $G$. 
Let $u$ be a non-singleton node of $x$ agents with the minimum ID agent $r_{min}$.
$DFS(r_{min})$ with $x$ initially co-located agents need to visit $x-1$ other empty nodes to settle all its co-located agents. The largest ID agent settles at $u$. When an empty node is visited by $DFS(r_{min})$ and it is unoccupied (Lemma \ref{lemma:confirmempty}),  an agent settles. If an unoccupied node is visited by the heads of two or more DFSs, an agent from one DFS settles. Since $k<n$, $DFS(r_{min})$ may need to traverse $|C(r_{min})|\Delta$ edges to settle its agents.  
Regarding memory, an agent settled at a node may need to store the information about $\ell_{C(r_{min})}$ DFSs that form the component $C(r_{min})$. Since each DFS needs  $O(\log (k+\Delta))$ bits at each node, the total memory needed at an agent at component $C(r_{min})$  is $O(\ell_{C(r_{min})} \log (k+\Delta))$ bits.  
%
%
%
%
\end{proof}


\begin{lemma}
\label{lemma:dispersion-election}
    If non-singleton nodes in the initial configuration are $\ell\geq 1$, $\ell_i\geq 1$ agents become local leaders in component $C_i$, where $\ell_i$ are the non-singleton nodes that belongs to $C_i$ after all agents disperse.
    The non-singleton local leader election finishes for each agent in $O(|C_{max}|\Delta)$ rounds. 
\end{lemma}
\begin{proof}
Consider first the case of $\ell=1$. A single $DFS(r_{min})$ runs forming a single component $C$.  When $DFS(r_{min})$ finishes, $r_{min}$ becomes a local leader. 

Consider now the case of $\ell\geq 2$.  There may be the case that the heads of two DFSs $DFS(r_1)$ and $DFS(r_2)$ meet. In this case, $DFS(r_1)$ continues, if $r_1.ID>r_2.ID$, otherwise $DFS(r_2)$. The DFS that stops will hand over remaining agents to be settled to the winning DFS. Additionally, no agent becomes local leader from the stopped DFS. Therefore, for two DFSs at least an agent from one DFS becomes a local leader. Therefore, if $\ell_i$ DFSs from $\ell_i$ non-singleton nodes meet to form a component $C_i$, there is at least one local leader elected.

We now prove the time complexity.
For $\ell=1$, an agent becomes a local leader in $O(C \cdot \Delta)=O(k\Delta)$ rounds (Lemma \ref{lemma:dispersion}), as $|C|=k$. For $\ell\geq 2$, the election finishes in $O(C_i \cdot \Delta)$ rounds in each component $C_i$.  
Therefore, for each agent,   the non-singleton local leader election finishes in $O(|C_{max}|\Delta)$ rounds.
\end{proof}

\begin{lemma}
\label{lemma:return-home-node}
Consider a local leader $\psi(u)$ that cannot become a global leader. $\psi(u)$ can return to $home(\psi(u))$.  
\end{lemma} 
\begin{proof}
Consider a local leader $\psi(u)$ that becomes $non\_candidate$ before running $Global\_Election$. That local leader is already at its $home(\psi(u))$. Now consider a local leader $\psi(u)$ that can run $Global\_Election$.
During $Global\_Election$ (Algorithm \ref{algorithm:leader-election}), $\psi(u)$ runs $DFS(roundNo_u,\psi(u))$. $DFS(roundNo_u,\psi(u))$ builds a DFS tree $T_u$ with its root $home(\psi(u))$. In $T_u$, there is a sequence of parent pointers from the head of $DFS(roundNo_u,\psi(u))$) to $home(\psi(u))$. Since every local leader runs its own $DFS(roundNo_u,\psi(u))$ and maintains $T_u$, $\psi(u)$ can follow the parent pointers in $T_u$ until reaching $home(\psi(u))$. 
\end{proof}

\begin{lemma}
\label{lemma:oscillation}
Consider an agent $\psi(u)$ that became local leader at $roundNo_u$. Let $\psi(u)$ belongs to component $C_i$. Using procedure $Oscillation(\psi(u))$ (Algorithm \ref{algorithm:oscillation}), $\psi(u)$ either becomes $non\_candidate$ or starts running $Global\_Election$ within  $O((|C_i|+\log^2k)\Delta)$ rounds. 
\end{lemma}
\begin{proof}
We look into whether $\psi(u)$ became local leader through $Singleton\_Election$ or $Non\_Singleton\_Election$.

We first consider the $Singleton\_Election$ case.
Suppose $\psi(u)$ became local leader at $roundNo_u$ and belongs to $C_i$. Let $\psi(w)$ be the agent at $w\in N(u)$ that is oscillating. $\psi(w)$ becomes non-oscillating by $roundNo_u+O(|C_i|\Delta)$. 
Additionally, each agent becoming local leader through $Singleton\_Election$ will be elected local leaders by $O(\Delta\log^2k)$ rounds, i.e., $roundNo_u\leq O(\Delta\log^2k)$ (Lemma \ref{lem:rounds_with_same_degree_leader}).
Therefore, by $roundNo_u+O(|C_i|\Delta)\leq O(\Delta\log^2 k)+O(|C_i|\Delta)$ rounds, either $\psi(u)$ becomes a non-candidate or starts running $Global\_Election$. 

We now consider the $Non\_Singleton\_Election$ case. Suppose $\psi(u)$ became local leader at $roundNo_u$ and belongs to $C_i$. Let $\psi(w)$ be the agent at the parent $w$ of node $u$ in the DFS tree built by $\psi(u)$. 
Let $\psi(w')$ be the agent at $w'\neq w\neq u$ oscillating to $w$. $\psi(w')$ becomes non-oscillating by $roundNo_u+O(|C_i|\Delta)$. Additionally, $roundNo_u\leq O(|C_i|\Delta$.  Therefore, by $roundNo_u+O(|C_i|\Delta)\leq O(|C_i|\Delta)$ rounds, either $\psi(u)$ becomes a non-candidate or starts running $Global\_Election$.

We have the claimed bound, combining the two results.
\end{proof}

\begin{lemma}
\label{lemma: leader-election}
Global\_Election (Algorithm \ref{algorithm:leader-election}) at each component $C_i$ elects a unique global leader at $C_i$ and terminates in $O((|C_i|+\log^2k)\Delta)$ rounds with $O(\zeta_i \log (k+\Delta))$ bits at each agent. 
\end{lemma}
\begin{proof}
Consider an agent $\psi(u)$ that becomes a local leader at node $u$ of component $C_i$. It either becomes non-candidate or initiates $Global\_Election(\psi(u))$ (Algorithm \ref{algorithm:leader-election}) by round $O((|C_i|+\log^2k)\Delta)$. $Global\_Election(\psi(u))$  runs $DFS(roundNo_u,\psi(u))$.
Suppose $DFS(roundNo_u,\psi(u))$ visits an empty node $w$. We have from Lemma \ref{lemma:confirmempty} that it can be confirmed in an additional round whether $w$ is in fact unoccupied or occupied running Algorithm \ref{algorithm:confirm-empty}.  If it is unoccupied, then $DFS(roundNo_u,\psi(u))$ simply backtracks to the previous node since $w$ does not belong to $C_i$. If $DFS(roundNo_u,\psi(u))$ meets another $DFS(roundNo_v,\psi(v))$ then either $DFS(roundNo_u,\psi(u))$ or $DFS(roundNo_u,\psi(v))$ continues. If $DFS(roundNo_u,\psi(u))$ meets $DFS(r)$ doing $Non\_Singleton\_Election$ then $DFS(roundNo_u,\psi(u))$ stops as agent $r$ is eligible to become a local leader at $C_i$ later in time and which will run $Global\_Election$ to become a global leader in $C_i$.  
Therefore, $Global\_Election$ elects a unique global leader in $C_i$.  
The total time to finish $Global\_Election$ (Algorithm \ref{algorithm:leader-election}) in $C_i$ is  $O(|C_i|\Delta)$ rounds. To store information about $Global\_Election$ for each local leader in $C_i$, $O(\log(k+\Delta))$ bits at each agent is sufficient. Given $\zeta_i$ local leaders in $C_i$, the total memory requirement is $O(\zeta_i \log (k+\Delta))$ bits at each agent. 
\end{proof}

\begin{lemma}
\label{lemma:component-hop}
    After $Global\_Election$ finishes, for any two agents $\psi(u),\psi(v)$ such that $u\in C_i$ and $v\in C_j$, $hop(u,v)\geq 2$.  
\end{lemma}
\begin{proof}
    We prove this by contradiction. Suppose $hop(u,v)=1$, meaning that agents $\psi(u)$ and $\psi(v)$ are neighbors but belong to two different components. Let $u\in C_i$ and $v\in C_j$. For this to happen, the $Global\_Election$ in $C_i$ ($C_j$) must not visit node $v$ ($u$). Consider the $Global\_Election$ DFS on either $C_i$ or $C_j$ that started last in time (suppose that is from $C_i$).  That $Global\_Election$ DFS visits one-hop neighborhood of all the nodes in $C_i$. Since $hop(u,v)=1$,  $Global\_Election$ DFS must visit node $v$ and since $v$ has agent $\psi(v)$, $\psi(v)$ is considered as the agent that belongs to component $C_i$. Hence, a contradiction. 
\end{proof}

\begin{theorem}[{\bf Stabilizing Algorithm, $k<n$}]
\label{theorem:stabilizing}
There is a stabilizing deterministic algorithm for $k<n$ agents in the agentic model that elects one agent as a leader in each component $C_i$, without agents knowing any graph parameter a priori.  
The time complexity of the algorithm is $O((|C_{max}|+\log^2k)\Delta)$ rounds and the memory complexity is $O(\max\{\ell_{\max},\zeta_{\max},\log (k+\Delta)\} \log (k+\Delta))$ bits per agent. 
\end{theorem}

\begin{proof}
Consider the dispersed initial configuration. In Stage 1, only $Singleton\_Election$ runs. We have from Lemma \ref{lemma:dispersed-election}, that at least one agent becomes a local leader. We have from Lemma \ref{lem:rounds_with_same_degree_leader}, $Singleton\_Election$ for each agent finishes in $O(\Delta\log^2k)$ rounds with memory $O(\log^2(k+\Delta))$ bits per agent. Let $\zeta_i$ be the number of local leaders in each $C_i$. In Stage 2, $Global\_Election$ for each local leader finishes in $O(|C_{max}|\Delta)$ rounds with $O(\zeta_{max} \log(k+\Delta))$ bits at each agent.  Therefore, time complexity becomes 
$O((|C_{max}|+\log^2k)\Delta)$ rounds and the memory complexity becomes $O(\max\{\zeta_{\max},\log (k+\Delta)\} \log (k+\Delta))$ bits per agent. 

Consider the general initial configuration. We have two sub-cases. (1) There is no singleton node. (2) There is a mix of singleton and non-singleton nodes. Consider no singleton node case first.  
$\ell$ instances of $Non\_Singleton\_Election$ finish in $O(|C_{max}| \Delta)$ rounds in each component.  The memory is $O(\ell_{max} \log(k+\Delta))$ bits per agent. $Global\_Election$ then finishes in next $O(|C_{max}| \Delta)$ rounds in each component with memory $O(\zeta_{\max} \log(k+\Delta))$ bits per agent (Lemma \ref{lemma: leader-election}).  Consider now the case of $\ell$ non-singleton nodes and at least a singleton node. $Singleton\_Election$ finishes in $O(\Delta\log^2k)$ rounds (Lemma \ref{lem:rounds_with_same_degree_leader}). Everything else remains the same.



Consider the local leaders 
    that cannot become global leaders after running $Global\_Election$. We have from Lemma \ref{lemma:return-home-node} that they can return to their home nodes. Notice that returning to home nodes takes $O(|C_{max}|)$ rounds since agents can traverse parent pointers on the DFS tree they build while running $Global\_Election$.  The memory remains $O(\zeta_{\max}\log(k+\Delta))$ bits per agent. Lemma \ref{lemma:component-hop} shows that the agents in two components are at least 2-hop away from each other, satisfying our component definition. Therefore, our leader election is correct that one component has exactly one leader. 
    
Combining all these time and memory bounds, we have time complexity $O((|C_{max}|+\log^2k)\Delta)$ and memory complexity $O(\max\{\ell_{\max},\zeta_{\max},\log (k+\Delta)\} \log (k+\Delta))$ bits per agent.
\end{proof}

\subsection{Explicit Algorithm}
We now discuss how Algorithm \ref{algorithm:leader_election} can be made explicit (no overtaking of global leader in each component $C_i$) with agents knowing $k,\Delta$. Our explicit algorithm follows the 2-stage approach as in the stabilizing algorithm. In Stage 1, singleton agents run  {\em Singleton\_Election} and non-singleton agents run {\em Non\_Singleton\_Election} to become local leaders. The local leaders then run {\em Global\_Election} in Stage 2. 
Stage 1 runs for $c_1\cdot k\Delta$ rounds, i.e., the agents that become local leaders in Stage 1 start Stage 2 at round $c_1\cdot k\Delta+1$.
Stage 2 also runs for $c_1\cdot k\Delta$ rounds, i.e., a local leader agent eligible to become a global leader declares itself as a global leader after Stage 2 runs for $c_1\cdot k\Delta$ rounds. 
We prove that this approach avoids  overtaking, meaning once one local leader agent becomes a global leader, there is no other local leader that becomes global leader later in time.

We now discuss how $Singleton\_Election(\psi(u))$ executes knowing $k,\Delta$. $\psi(u)$ meets all its neighboring singleton agents in $O(\Delta\log k)$ rounds (see discussion in message-passing model simulation in Section \ref{section:introduction}). Therefore,  $Singleton\_Election(\psi(u))$  finishes in $O(\Delta \log k)$ rounds. $Non\_Singleton\_Election(\psi(u))$ finishes in $O(k\Delta)$. The eligible agents to become local leader declare themselves as local leaders at round $c_1\cdot k\Delta$ (not before). 
Additionally, at $c_1\cdot k\Delta$, all agents have already dispersed. Let $C_1,C_2,\ldots,C_{\kappa}$ be $\kappa$ components formed at round $c_1\cdot k\Delta$.  It is easy to see that components remain fixed after $c_1\cdot k\Delta$.

The local leaders then run $Oscillation$ for $c_1\cdot k\Delta$ rounds, some of them become $non\_candidate$ and some of them are eligible to run  $Global\_Election$. The local leader $\psi(u)$ eligible to run $Global\_Election$ then runs  $DFS(\psi(u))$. We do not need $roundNo$ here since all eligible local leaders run $Global\_Election$ at the same round and symmetry breaking based on the leader ID is enough. When $DFS(\psi(u)$ meets $DFS(\psi(v))$, $DFS(\psi(v))$ stops if $\psi(u).ID>\psi(v).ID$. If $\psi(u)$ finishes visiting all the nodes in $C(\psi(u))$, then it can declare itself as a global leader of the component $C(\psi(u))$. Notice that if $\psi(u)$ visits all the nodes of $C(\psi(u))$ then it must be the case that it got the highest priority among all other local leaders in $C(\psi(u))$ that ran $Global\_Election$ and hence $DFS(\psi(u))$ stopped all of them.

\begin{theorem}[{\bf Explicit Algorithm, $k<n$}]
\label{theorem:explicitkn}
There is an explicit deterministic algorithm for $k<n$ agents in the agentic model that elects one agent as a leader in each component $C_i$ with agents knowing $k,\Delta$ a priori. The time complexity of the algorithm is $O(k\Delta)$ rounds and the memory complexity is  $O(\max\{\ell_{max},\zeta_{\max}\} \log (k+\Delta))$ bits per agent.
\end{theorem}
\begin{proof}
The only difference with Theorem \ref{theorem:stabilizing} for the stabilizing algorithm is that since $k,\Delta$ are known, $Singleton\_Election$ finishes in $O(\Delta \log k)$ rounds with only $O(\log(k+\Delta))$ bits per agent. The total time complexity of Stage 1 and Stage 2 is $O(k\Delta)$. Therefore, the time complexity becomes $O(k\Delta)$ rounds and the memory complexity becomes $O(\max\{\ell_{max},\zeta_{\max}\} \log (k+\Delta))$ bits per agent. The correctness proof of two agents in two components are at least 2 hops apart follows easily from the proof of Lemma \ref{lemma:component-hop}. 
\end{proof}

\section{Leader Election, $k=n$} 
\label{section:leader}
We discuss here how the stabilizing algorithm for the case of $k<n$ (not knowing $k,\Delta$) can be made explicit for the case of $k=n$ without knowing $k,\Delta$.
Additionally, only one global leader will be elected (i.e., the number of components $\kappa=1$). The time complexity is $O(m)$ and memory complexity is $O(\max\{\ell,\zeta,\log n\} \log n)$ bits per agent with $\ell$ being the number of non-singleton nodes in the initial configuration and $\zeta$ being the number of nodes that become local leaders in Stage 1.

We discuss here where we modify Algorithm \ref{algorithm:leader_election} for the case of $k=n$. The modified algorithm still has two stages, Stage 1 and Stage 2. In Stage 1,  $Non\_Singleton\_Election$ procedure execute as in Algorithm \ref{algorithm:dispersionk}. 
The change is on $Singleton\_Election$ and $Global\_Election$.
$Singleton\_Election$ elects an agent $\psi(u)$ at node $u$ as a local leader if and only if $\psi(u)$ finds all neighbors in $N(u)$ have a singleton agent positioned. 
Regarding $Global\_Election$, 
the $DFS(roundNo_u,\psi(u))$ run by local leader agent $\psi(u)$ becomes a global leader if and only if it can visit all the edges of $G$. That is, if it visits an unoccupied node while running $DFS(roundNo_u,\psi(u))$, then $\psi(u)$ stops  $DFS(roundNo_u,\psi(u))$, becomes $non\_candidate$, and returns to $home(\psi(u))$. This is because the encounter of an unoccupied node by $Global\_Election$ means that at least one $Non\_Singleton\_Election(r_v)$ procedure by an initially non-singleton agent has not been finished yet. $Confirm\_Empty()$ (Algorithm \ref{algorithm:confirm-empty}) can correctly guarantee that whether an empty node visited is occupied or unoccupied since for the occupied node there is an oscillating agent visiting it every two rounds.  

\begin{theorem}[{\bf Explicit Algorithm, $k=n$}]
\label{theorem:explicit}
There is an explicit deterministic algorithm for $k=n$ agents in the agentic model that elects one agent as a leader in the graph $G$, without agents knowing any graph parameter a priori. The time complexity of the algorithm is $O(m)$ rounds and the memory complexity is $O(\max\{\ell,\zeta,\log n\} \log n)$ bits per agent. %
\end{theorem}
\begin{proof}
This algorithm follows the stabilizing algorithm for $k<n$ but avoids overtaking. 
Stage 1 finishes in $O(m)$ rounds as the DFS may need to visit all the edges of the graph before $Non\_Singleton\_Election$ settles all the agents. $Singleton\_Election$ finishes in $O(\Delta\log^2n)$ rounds. Stage 2 also finishes in $O(m)$ rounds. A local leader that cannot become a global leader can return its home node in $O(n)$ rounds following the path in the DFS tree built during $Global\_Election$. Therefore, the total time complexity is $O(m)$ rounds.

Regarding memory complexity, $Neighbor\_Exploration\_with\_Padding()$ needs $O(\log^2n)$ bits per agent. $Non\_Singleton\_Election$ needs $O(\max\{\ell,\zeta\} \log n)$ bits per agent, with $\zeta$ being the number of local leaders in Stage 1. $Global\_Election$ also needs $O(\max\{\ell,\zeta\} \log n)$ bits per agent to facilitate local leaders that become $non\_candidate$ during $Global\_Election$ to return to their home nodes. Therefore, the total memory complexity is $O(\max\{\ell,\zeta,\log n\} \log n)$ bits per agent.
\end{proof}

\section{MST Construction}\label{sec: MST construction}
In this section, we present a stabilizing deterministic algorithm for $k<n$ to construct an MST of each component $C_j$. We will discuss, towards the end of this section, an explicit version of this algorithm for $k<n$, knowing $k,\Delta$ and explicit version for $k=n$ not knowing any graph parameters.

We start the MST construction process after a leader is elected.  
Given a leader $r_l$ of Section \ref{section:leader} and its DFS tree $T_{r_l}$ built while running $DFS(roundNo_l,r_l)$, we discuss the algorithm assuming leader $r_l$ has already been stabilized, i.e., it will not be overtaken. We will discuss later how to handle the case of the leader that started MST construction being overtaken later.   
The agents do not need to know any graph parameters.
The MST construction finishes in $O(\Delta\log^2k+|C_j|\Delta+|C_j|\log|C_j|)$ rounds with $O(\max\{\ell_j,\zeta_j,\Delta,\log (k+\Delta)\} \log (k+\Delta))$ bits at each agent. We consider $G$ to be weighted and hence our MST for each component $C_j$ is a minimum weight MST of $C_j$. Having $\kappa$ components with $\kappa$ leaders, we will have $\kappa$ MSTs, one per component. 


After Algorithm \ref{algorithm:leader_election}, let $r_v$ be a leader of component $C_j$ positioned at node $v\in C_j$. 
Leader $r_v$ runs $DFS(r_v)$ to visit all other $|C_j|-1$ agents (at $|C_j|-1$ nodes) at component $C_j$ and assign them rank based on the order they are visited, i.e., the agent visited $i$-th  in the order receives rank $i$. This is done by $r_v$ revisiting the DFS tree $T_{r_v}$ built by $r_v$ while running $DFS(roundNo_v,r_v)$ to elect itself as a leader during Algorithm \ref{algorithm:leader-election} ($Global\_Election$). This revisit finishes in $O(|C_j|\Delta)$ rounds. 

As soon as an agent receives its rank, it considers itself as a single sub-component. 
The leader $r_v$ at node $v$ (which has rank-1) starts MST construction. $r_v$ includes the MOE (minimum weight outgoing edge) adjacent to $v$ in its sub-component and passes a token (message) to the rank-2 agent. This process runs iteratively until the rank-$(|C_j|-1)$ agent passes the token to the rank-$|C_j|$ agent. Consequently, the rank-$|C_j|$ agent includes in its sub-component the MOE available and passes the token to the rank-1 agent ($r_v$). This whole process of passing the token from rank-1 node to rank-$|C_j|$ node and back to rank-1 node is one phase. 

In the next phase, the minimum rank agent would include the MOE available to its sub-component and pass the token to the next minimum rank agent that belongs to another sub-component (within the component). In this way, the token reaches the minimum rank agent from the highest rank agent and this phase is completed. This process is repeated phase-by-phase until there is a single sub-component left, which is component $C_j$. Eventually, we have an MST of component $C_j$ with $|C_j|$ agents. Let us call this algorithm \textit{MST\_Construction}. 
A complete pseudocode is given in Algorithm~\ref{alg: mst_construction}. 

Our algorithm resembles the MST algorithm of Gallager, Humblet, and Spira \cite{Gallager83} with the difference that we start MST construction through ranks already provided to agents, whereas in \cite{Gallager83} all nodes have the same rank. The merging of two same rank sub-components with rank $k$ in \cite{Gallager83} provides a new rank of  $k+1$ for the merged sub-component. In ours, there will be no same rank sub-component, and hence the merged sub-component gets the rank of one of the sub-components merged.  


\begin{figure}[!t]
    \centering
    \begin{tikzpicture}[scale=1.3, transform shape, level distance=1.5cm,
        level 1/.style={sibling distance=3cm},
        level 2/.style={sibling distance=1.5cm}]
        \node[circle,draw] (root) {$r_y$}
          child {node[circle,draw] (child1) {$r_x$}
            child {node[circle,draw, red] (child2) {$r_w$}
              child {node[circle,draw] (child8) {$r_h$}}}
            child {node[circle,draw] (child3) {$r_e$}}}
          child {node[circle,draw] (child4) {$r_c$}
            child {node[circle,draw] (child5) {$r_f$}}
            child {node[circle,draw] (child6) {$r_g$}}};
        \draw[->] (root) -- (child1);
        \draw[->] (child1) -- (child2);
        \draw[->] (child1) -- (child3);
        \draw[->] (root) -- (child4);
        \draw[->] (child4) -- (child5);
        \draw[->] (child4) -- (child6);
        \draw[->] (child2) -- (child8);
        
        \begin{scope}[xshift=-5.5cm]
        \node[circle,draw] (root2) {$r_a$}
          child {node[circle,draw, red] (child42) {$r_u$}
            child {node[circle,draw] (child52) {$r_v$}}
            child {node[circle,draw] (child62) {$r_z$}}};
        \draw[->] (root2) -- (child42);
        \draw[->] (child42) -- (child52);
        \draw[->] (child42) -- (child62);
        \draw[dotted, blue] (child42) -- (child2) node[midway, above, font=\footnotesize, blue] {MOE};
        \end{scope}
    \end{tikzpicture}
    \caption{Sub-components $C_{r_w}$ and $C_{r_u}$ before merging.}
    \label{fig1}
\end{figure}
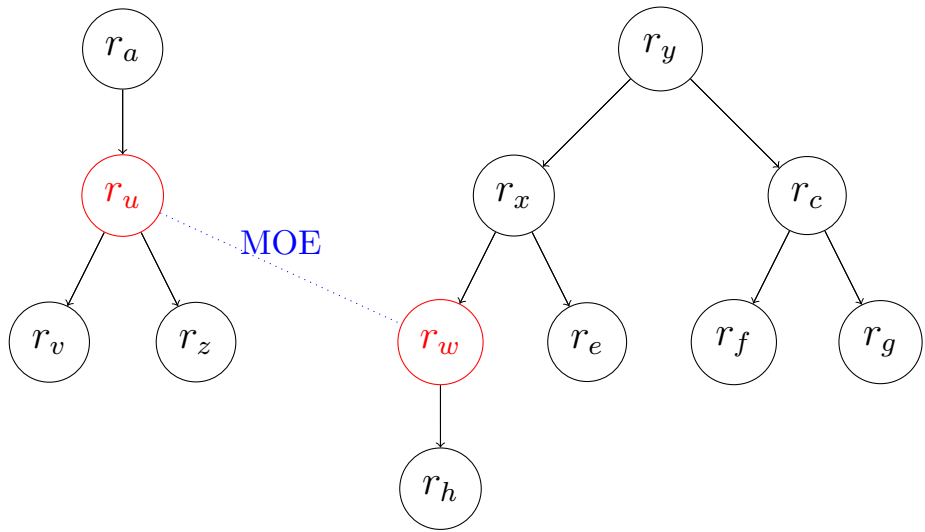

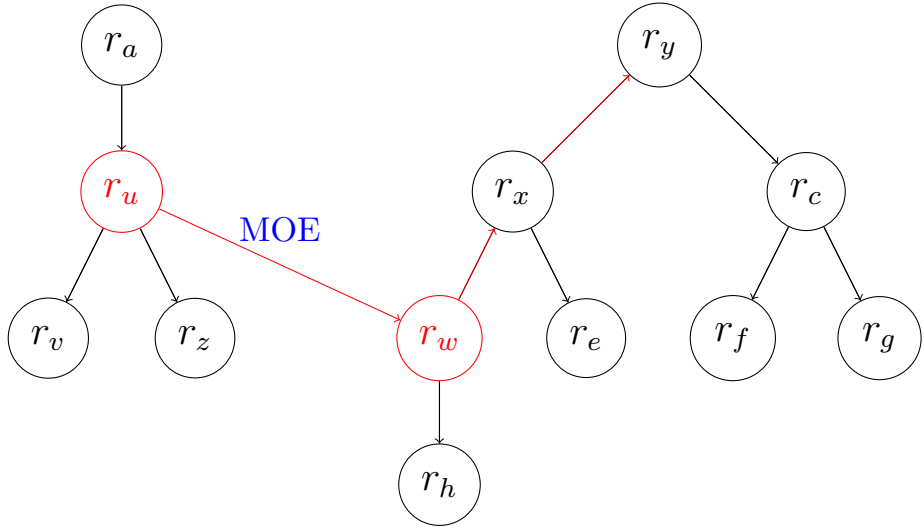
\begin{figure}[!h]
    \centering
    \begin{tikzpicture}[scale=1.3, transform shape, level distance=1.5cm,
        level 1/.style={sibling distance=3cm},
        level 2/.style={sibling distance=1.5cm}]
        \node[circle,draw] (root) {$r_y$}
          child {node[circle,draw] (child1) {$r_x$}
            child {node[circle,draw, red] (child2) {$r_w$}
              child {node[circle,draw] (child8) {$r_h$}}}
            child {node[circle,draw] (child3) {$r_e$}}}
          child {node[circle,draw] (child4) {$r_c$}
            child {node[circle,draw] (child5) {$r_f$}}
            child {node[circle,draw] (child6) {$r_g$}}};
        \draw[<-,red] (root) -- (child1);
        \draw[<-, red] (child1) -- (child2);
        \draw[->] (child1) -- (child3);
        \draw[->] (root) -- (child4);
        \draw[->] (child4) -- (child5);
        \draw[->] (child4) -- (child6);
        \draw[->] (child2) -- (child8);
        
        \begin{scope}[xshift=-5.5cm]
        \node[circle,draw] (root2) {$r_a$}
          child {node[circle,draw, red] (child42) {$r_u$}
            child {node[circle,draw] (child52) {$r_v$}}
            child {node[circle,draw] (child62) {$r_z$}}};
        \draw[->] (root2) -- (child42);
        \draw[->] (child42) -- (child52);
        \draw[->] (child42) -- (child62);
        \draw[->, red] (child42) -- (child2) node[midway, above = 0.1cm, font=\footnotesize, blue] {MOE};
        \end{scope}
    \end{tikzpicture}
    \caption{Sub-component $C^{new}_{r_u}$ after merging.}
    \label{fig2}
\end{figure}

\begin{algorithm}[h!]
\footnotesize{
\SetKwInput{KwInput}{Input}
\SetKwInput{KwEnsure}{Ensure}
\KwInput{An $n$-node anonymous network with $k\leq n$ agents with unique IDs placed on $k$ nodes forming $\kappa$ components with an agent $r_l$ at a node elected as a leader in each component $C_j$ (Algorithm \ref{algorithm:leader_election}) with DFS tree $T_{r_l}$  built while running $DFS(roundNo_l,r_l)$ (Algorithm \ref{algorithm:leader-election}).} 
\KwEnsure{MST construction of component $C_j$.}    

The leader assumes rank $1$.  It re-traverses the DFS tree $T_{r_l}$ and returns to its home node.  While re-traversing $T_{r_l}$, it provides the distinct rank from the range of $[2, |C_j|]$ to the agents at $|C_j|-1$ other nodes in order of visit. 
\label{line: rank_numbering}

Each agent $r_u\in C_j$ considers itself as a sub-component $C_{r_u}$ with $rank(C_{r_u})\leftarrow rank(r_u)$. 
\label{line: single_component}

The leader $r_l\in C_j$ generates a token.  Let its sub-component be $C_{r_l}$.
 
\While{$|C_{r_l}|<|C_j|$}
{   \If{$r_u$ has the token}
    {
        \If{$rank(r_u) \leq rank(C_{r_u})$}
        {
            Agent $r_u$ finds the MOE of the $C_{r_u}$ going to another sub-component $C_{r_w}$ connecting $r_u$ with $r_w$.

            \If{$rank(C_{r_u}) < rank(C_{r_w})$}
            {
                $r_w$ becomes the root node of $C_{r_w}$ by reversing the parent-child pointers from $r_w$ up to the root node of $C_{r_w}$ and $r_u$ becomes the parent of $r_w$ (Fig.~\ref{fig1}). $C_{r_w}$ then merges with $C_{r_u}$ giving a new sub-component $C_{r_u}^{new}$ with $rank(C^{new}_{r_u})\leftarrow rank(C_{r_u})$ (Fig.~\ref{fig2}).
            }
            
            \If{$rank(r_u) < |C_j|$}
            {
                $r_u$ passes the token to agent $r_v$ with  $rank(r_v)=rank(r_u)+1$ using $T_{r_l}$.
            }
        }
        \uElseIf{$rank(r_u) > rank(C_{r_u})$ and $rank(r_u) < |C_j|$}
        {
            $r_u$ passes the token to agent $r_v$ with  $rank(r_v)=rank(r_u)+1$ using $T_{r_l}$.
        }

        \ElseIf{$rank(r_u) = |C_j|$}
        {
        $r_u$ passes the token to the leader $r_l$ (with rank $1$)  using $T_{r_l}$. This token passing visits the parents of $r_u$ until the token reaches the root node of $T_{r_l}$, where the leader is positioned. 
          
        }
    }
} 
\caption{{$MST\_Construction$}}
\label{alg: mst_construction}

        
      
}
\end{algorithm}

Figs.~\ref{fig1} and \ref{fig2} illustrate these ideas. Fig.~\ref{fig1} shows sub-components $C_{r_u}$ and $C_{r_w}$ before they merge due to the MOE connecting $r_u$ with $r_w$. Fig.~\ref{fig2} captures the merged sub-component $C^{new}_{r_u}$ such that root of the  $C^{new}_{r_u}$ remains unchanged and the pointer changes occurred in the $C_{r_w}$ sub-component during the merging. The directed edge denotes new parent-child relationships.

We now discuss the overtaking issue of a leader that started MST construction in $C_j$. Suppose leader $r_l$ started MST construction and it is overtaken by another leader $r_l'$. In the MST construction, we ask leaders to provide nodes with $roundNo$ information at which they became leaders. If an agent $r$ now gets to know that a leader $r_l'$ with $roundNo_{l'}>roundNo_l$ started MST construction, $r$ discards its MST construction for $r_l$ and initiates the process for $r_l'$. Since leader election stabilizes to a single global leader in each component $C_i$, eventually, all MST construction processes from overtaken leaders stabilize to a single process from the leader that will not be overtaken. 

\vspace{2mm}
\noindent{\bf Analysis of the Algorithm.}
We now analyze Algorithm \ref{alg: mst_construction} for its correctness, time, and memory complexities. 

\begin{lemma}
\label{lemma:overtakeMST}
Suppose a leader $r_l$ running MST construction (Algorithm \ref{alg: mst_construction}) in component $C_j$ is overtaken by another leader $r_l'$. Algorithm \ref{alg: mst_construction} by $r_l'$ subsumes the MST construction of $r_l$.
\end{lemma}
\begin{proof}
The MST construction instance of $r_l'$ meets the MST construction instance of $r_l$. In every node, $r_l'$ will ask the agent to follow the MST construction instance of $r_l'$ discarding the information regarding $r_l$. Since agents in $C_i$ do not change their home nodes during MST construction, the subsumption is guaranteed.
\end{proof}

\begin{lemma}
    Algorithm~\ref{alg: mst_construction} 
    generates the MST of component $C_j$. \label{lem: mst_correctness}
\end{lemma}
\begin{proof}
We prove this in three steps: firstly, Algorithm~\ref{alg: mst_construction} constructs a tree in $C_j$; secondly,  the constructed tree is a spanning tree in $C_j$; finally, the spanning tree is, indeed, a minimum spanning tree in $C_j$. Firstly, we prove by contradiction that no cycle is generated during Algorithm~\ref{alg: mst_construction}. Let us suppose, there exists a cycle at any point during the algorithm. Then it implies two sub-components with the same rank merged at some point, which is a contradiction. Secondly, let us suppose there exist at least two sub-components at the end of the algorithm. This implies that the leader sub-component (rank-$1$) did not merge with the other sub-component and terminated the algorithm, which contradicts the fact that the algorithm terminates when the leader is connected to $n$ agents altogether.

Finally, consider that the tree formed for $C_j$ by our algorithm is $T$ and the MST is $T^{*}$. Note that in Algorithm~\ref{alg: mst_construction}, each edge added to the MST tree is selected by the selection of an MOE.  If $T = T^{*}$ then $T$ is minimum spanning tree. If $T \neq T^{*}$ then there exists an edge $e \in T^{*}$ of minimum weight such that $e \notin T$. Therefore, there exists a phase in which $e$ was not considered during sub-component merging and an edge with higher weight, say $e'$, was considered. But this is contradictory to Algorithm~\ref{alg: mst_construction} which merges the sub-components with an MOE. Therefore, Algorithm~\ref{alg: mst_construction} constructs the MST of $C_j$.
\end{proof}



\begin{lemma}\label{lem: phase_number}
    In Algorithm~\ref{alg: mst_construction}, 
    the leader initiates the merging $O(\log |C_j|)$ times in component $C_j$.
\end{lemma}
\begin{proof}
    Initially, there are $|C_j|$ single-node sub-components in Algorithm~\ref{alg: mst_construction} (Line~\ref{line: single_component}) for each component $C_j$. In each phase (leader initiating a token until the token returns to the leader), each sub-component merges at least once with another sub-component. Therefore, after every phase, the number of sub-components is reduced by at least half. Consequently, after $O(\log n)$ phases, there remains only a single sub-component of $|C_j|$ nodes, which is component $C_j$ itself. 
\end{proof}


\begin{theorem}[{\bf Stabilizing MST, $k<n$}]
\label{theorem:stabilizingMST}
There is a stabilizing deterministic algorithm for constructing MST for each component $C_j$ in the agentic model, without agents knowing any graph parameter a priori.  The time complexity of the algorithm is $O(\Delta\log^2k+|C_{\max}|(\Delta+\log |C_{\max}|))$ rounds and the memory complexity is 
 $O(\max\{\ell_{\max},\zeta_{\max},\Delta,\log (k+\Delta)\} \log (k+\Delta))$ bits per agent.
\end{theorem}

\begin{proof}
Providing ranks to agents finishes in $O(|C_j|\Delta)$ rounds for each component $C_j$. The while loop (Algorithm~\ref{alg: mst_construction}) performs two operations -- token passing and merging.
In token passing, the token is passed through the edge of the tree $T_{DFS}$, and an edge is not traversed more than twice. Therefore, in a phase, to pass the token from rank-$1$ agent to rank-$|C_j|$ agent takes $O(|C_j|)$ rounds. From rank-$n$ agent the token returns to the leader again in $O(|C_j|)$ rounds. Combining this with Lemma~\ref{lem: phase_number}, token passing takes $O(|C_j| \log |C_j|)$ rounds. 

In the process of merging, an agent $r_u$ visits at most three types of edges: i) MOE edges within its sub-component $C_{r_u}$ ii) edges traversed to find MOE iii) reversing the edges from $r_w$ until the root of $C_{r_w}$ when it merges with another sub-component $C_{r_u}$ at $r_w$. In the case of i) MOE is the part of the sub-component $C_{r_u}$, i.e., a tree. Its traversal finishes in $O(|C_{r_u}|)$ rounds. In a phase, the combined size of all the sub-components is $O(|C_j|)$. In case ii) edges that are traversed to find the MOE were either part of MOE or not. In case they became part of MOE, they were traversed two times. There are at most $(|C_j|-1)$ such edges throughout the process. On the other hand, if some edges did not become part of the MOE then they were never traversed again. Therefore, there are in total $|C_j|\Delta-(|C_j|-1)$ such edges. In case iii) reversing an edge from its merging point to the root can not be more than its sub-component size. Therefore, reversing of edge for agent $r_u$ takes $O(|C_{r_u}|)$. Combining the time for the cases i) and iii) per phase with $O(\log |C_j|)$ phases (Lemma~\ref{lem: phase_number}), we have total runtime  $O(|C_j| \log |C_j|)$ rounds and for case ii) we have total $O(|C_j|\Delta)$ rounds throughout the execution. Thus, the overall round complexity becomes $O(|C_j|\Delta + |C_j| \log |C_j|)$ for MST construction. Combining this with the leader election time complexity of $O((|C_j|+\log^2k)\Delta)$ (Theorem \ref{theorem:stabilizing}), we have the claimed time bound. 


For memory, rank numbering takes  $O(\log k)$ bits at each agent to re-traverse the DFS tree (constructed during Algorithm \ref{algorithm:leader-election}). Furthermore, each agent stores $O(\log k)$ bits to keep the account of the ID/rank and sub-component rank. Also, there might be a case in which all the neighbors are part of the MST. Therefore, in the worst case, the highest degree ($\Delta$) agent (agent placed at the highest degree node) keeps the account of all the MST edges and requires $O(\Delta \log k)$ memory. Hence, the overall memory per agent in Algorithm~\ref{alg: mst_construction} is $O(\Delta \log k)$ bits. Combining this time/memory with the time/memory needed for leader election (Theorem \ref{theorem:stabilizing}), we have the claimed memory complexity bound.
\end{proof}

\noindent{\bf Explicit MST, $k<n$.} The MST construction starts after a leader is elected through 
the explicit leader election algorithm,
i.e., after $O(k\Delta)$ rounds. Due to explicit leader election, when MST construction starts, there is a unique global leader in each component $C_j$. The MST construction at $C_j$ then finishes in $O(|C_j|\Delta + |C_j| \log |C_j|)$ rounds with memory $O(\Delta \log k)$ bits. Therefore, combining time/memory bounds for leader election (Theorem \ref{theorem:explicitkn}), the total time complexity becomes $O(k\Delta+|C_j|\Delta + |C_j| \log |C_j|)=O(k\Delta)$ rounds and the memory complexity becomes $O(\max\{\ell_{\max},\zeta_{\max},\Delta\} \log (k+\Delta))$ bits per agent. 

\begin{theorem}[{\bf Explicit MST, $k<n$}]
\label{theorem:explicitMST}
There is an explicit deterministic algorithm for constructing MST for each component $C_j$ in the agentic model, with agents knowing $k,\Delta$ a priori.  The time complexity of the algorithm is $O(k\Delta)$ rounds and the memory complexity is 
 $O(\max\{\ell_{\max},\zeta_{\max},\Delta\} \log (k+\Delta))$ bits per agent.
\end{theorem}

\noindent{\bf Explicit MST, $k=n$.} The MST construction starts after a leader is elected, i.e., after $O(m)$ rounds. Due to explicit leader election, when MST construction starts, there is a unique global leader in the single component $C$ which is $G$. The MST construction at $C$ then finishes in $O(|C_j|\Delta + |C_j| \log |C_j|)\leq  O(m + n \log n)$ rounds with memory $O(\Delta \log n)$ bits; when whole graph is a single component, then $|C_j|\Delta$ is an overestimate on time bound since leader election and MST construction both finish after visiting all edges, i.e., $|C_j|\Delta$ can be replaced with $m$.  Therefore, combining time/memory bounds for leader election (Theorem \ref{theorem:explicit}), the total time complexity becomes $O(m+n\log n)$ rounds and the memory complexity becomes $O(\max\{\ell,\zeta,\Delta,\log n\} \log n)$ bits per agent. 

\begin{theorem}[{\bf Explicit MST, $k=n$}]
\label{theorem:explicitMSTkn}
There is an explicit deterministic algorithm for constructing MST of $G$ in the agentic model when $k=n$, without agents knowing any graph parameter a priori.  The time complexity of the algorithm is $O(m+n\log n)$ rounds and the memory complexity is 
 $O(\max\{\ell,\zeta,\Delta,\log n\} \log n)$ bits per agent.
\end{theorem}

\section{Concluding Remarks}
\label{section:conclusion}
We have studied two fundamental distributed graph problems, leader election, and MST, in the agentic model. The agentic model poses unique challenges compared to the well-studied message-passing model. We provided stabilizing as well as explicit deterministic algorithms for leader election when $k<n$. The stabilizing algorithms do not require agents to know graph parameters but overtaking of a leader may occur until a certain time. The explicit algorithms avoid overtaking but agents require knowledge of $k,\Delta$. For the case of $k=n$, we showed that the algorithm for leader election can be made explicit (no overtaking) without knowing $k,\Delta$. We then developed respective algorithms for MST for both $k<n$ and $k=n$.  To the best of our knowledge, both problems were studied for $k<n$ for the very first time in this paper. 

For future work, it would be interesting to improve the time and/or memory complexities of our leader election and MST algorithms.  It would be interesting to see where overtaking can be avoided without parameter knowledge for $k<n$. It would also be interesting to solve other fundamental distributed graph problems, such as maximal independent set, maximal dominating set, coloring, maximal matching, and minimum cut, in the agentic model for $k\leq n$. 
Some fundamental problems, such as the maximal independent set (MIS) and maximal dominating set (MDS), have been considered in the agentic model for $k=n$ \cite{PattanayakBCM24,ChandMS23}. The developed algorithms gather the robots to a single node after leader election and then present a technique to solve the problems. For the case of $k<n$, this idea does not extend well since some components may still be in leader election phase and some components may progress toward computing MIS and MDS and this creates a challenge on how to synchronize the agents that are in different phases. Finally, it would also be interesting to consider agent faults (crash or Byzantine) and be able to solve graph level tasks with non-faulty agents. 

\bibliographystyle{elsarticle-num}
\bibliography{references}

\begin{thebibliography}{10}
\expandafter\ifx\csname url\endcsname\relax
  \def\url#1{\texttt{#1}}\fi
\expandafter\ifx\csname urlprefix\endcsname\relax\def\urlprefix{URL }\fi
\expandafter\ifx\csname href\endcsname\relax
  \def\href#1#2{#2} \def\path#1{#1}\fi

\bibitem{KKMS24}
A.~D. Kshemkalyani, M.~Kumar, A.~R. Molla, G.~Sharma, Brief announcement: Agent-based leader election, mst, and beyond, in: DISC, Vol. 319 of LIPIcs, 2024, pp. 50:1--50:7.
\newblock \href {https://doi.org/10.4230/LIPICS.DISC.2024.50} {\path{doi:10.4230/LIPICS.DISC.2024.50}}.

\bibitem{Cong2021}
Y.~Cong, g.~Changjun, T.~Zhang, Y.~Gao, Underwater robot sensing technology: A survey, Fundamental Research 1 (03 2021).
\newblock \href {https://doi.org/10.1016/j.fmre.2021.03.002} {\path{doi:10.1016/j.fmre.2021.03.002}}.

\bibitem{LEE2018}
J.~Lee, S.~Shin, M.~Park, C.~Kim, Agent-based simulation and its application to analyze combat effectiveness in network-centric warfare considering communication failure environments, Mathematical Problems in Engineering 2018 (2018) 1--9.
\newblock \href {https://doi.org/10.1155/2018/2730671} {\path{doi:10.1155/2018/2730671}}.

\bibitem{Zhuge2018}
C.~Zhuge, C.~Shao, B.~Wei, An agent-based spatial urban social network generator: A case study of beijing, china, Journal of Computational Science 29 (09 2018).
\newblock \href {https://doi.org/10.1016/j.jocs.2018.09.005} {\path{doi:10.1016/j.jocs.2018.09.005}}.

\bibitem{ELSAYED2012}
A.~El-Sayed, P.~Scarborough, L.~Seemann, S.~Galea, Social network analysis and agent-based modeling in social epidemiology, Epidemiologic perspectives \& innovations : EP+I 9 (2012) 1.
\newblock \href {https://doi.org/10.1186/1742-5573-9-1} {\path{doi:10.1186/1742-5573-9-1}}.

\bibitem{Pandurangan0S18}
G.~Pandurangan, P.~Robinson, M.~Scquizzato, \href{https://doi.org/10.1145/3210377.3210409}{On the distributed complexity of large-scale graph computations}, in: C.~Scheideler, J.~T. Fineman (Eds.), SPAA, {ACM}, 2018, pp. 405--414.
\newblock \href {https://doi.org/10.1145/3210377.3210409} {\path{doi:10.1145/3210377.3210409}}.
\newline\urlprefix\url{https://doi.org/10.1145/3210377.3210409}

\bibitem{kshemkalyani2024agent}
A.~D. Kshemkalyani, M.~Kumar, A.~R. Molla, G.~Sharma, Agent-based mst construction, in: arXiv, 2024.
\newblock \href {http://arxiv.org/abs/2403.13716} {\path{arXiv:2403.13716}}.

\bibitem{KshemkalyaniAAMAS25}
A.~D. Kshemkalyani, M.~Kumar, A.~R. Molla, G.~Sharma, Near-linear time leader election in multiagent networks, in: AAMAS, AAMAS '25, 2025, p. 1218–1226.

\bibitem{KshemkalyaniICDCIT25}
A.~D. Kshemkalyani, M.~Kumar, A.~R. Molla, G.~Sharma, Faster leader election and its applications for mobile agents with parameter advice, in: ICDCIT.

\bibitem{Augustine:2018}
J.~Augustine, W.~K.~M. Jr., Dispersion of mobile robots: {A} study of memory-time trade-offs, in: ICDCN, 2018, pp. 1:1--1:10.

\bibitem{GarayKP93}
J.~A. Garay, S.~Kutten, D.~Peleg, \href{https://doi.org/10.1109/SFCS.1993.366821}{A sub-linear time distributed algorithm for minimum-weight spanning trees (extended abstract)}, in: FOCS, {IEEE} Computer Society, 1993, pp. 659--668.
\newblock \href {https://doi.org/10.1109/SFCS.1993.366821} {\path{doi:10.1109/SFCS.1993.366821}}.
\newline\urlprefix\url{https://doi.org/10.1109/SFCS.1993.366821}

\bibitem{Peleg90L}
D.~Peleg, \href{https://doi.org/10.1016/0743-7315(90)90074-Y}{Time-optimal leader election in general networks}, J. Parallel Distrib. Comput. 8~(1) (1990) 96–99.
\newblock \href {https://doi.org/10.1016/0743-7315(90)90074-Y} {\path{doi:10.1016/0743-7315(90)90074-Y}}.
\newline\urlprefix\url{https://doi.org/10.1016/0743-7315(90)90074-Y}

\bibitem{Kshemkalyani2025}
A.~D. Kshemkalyani, M.~Kumar, A.~R. Molla, D.~Pattanayak, G.~Sharma, Dispersion is (almost) optimal under (a)synchrony, in: SPAA, 2025.

\bibitem{Lann77}
G.~L. Lann, Distributed systems - towards a formal approach, in: B.~Gilchrist (Ed.), Information Processing, Proceedings of the 7th {IFIP} Congress 1977, Toronto, Canada, August 8-12, 1977, North-Holland, 1977, pp. 155--160.

\bibitem{Gallager83}
R.~G. Gallager, P.~A. Humblet, P.~M. Spira, \href{https://doi.org/10.1145/357195.357200}{A distributed algorithm for minimum-weight spanning trees}, ACM Trans. Program. Lang. Syst. 5~(1) (1983) 66–77.
\newblock \href {https://doi.org/10.1145/357195.357200} {\path{doi:10.1145/357195.357200}}.
\newline\urlprefix\url{https://doi.org/10.1145/357195.357200}

\bibitem{Awerbuch87}
B.~Awerbuch, \href{https://doi.org/10.1145/28395.28421}{Optimal distributed algorithms for minimum weight spanning tree, counting, leader election and related problems (detailed summary)}, in: A.~V. Aho (Ed.), STOC, {ACM}, 1987, pp. 230--240.
\newblock \href {https://doi.org/10.1145/28395.28421} {\path{doi:10.1145/28395.28421}}.
\newline\urlprefix\url{https://doi.org/10.1145/28395.28421}

\bibitem{KPP0T15}
S.~Kutten, G.~Pandurangan, D.~Peleg, P.~Robinson, A.~Trehan, On the complexity of universal leader election, J. {ACM} 62~(1) (2015) 7:1--7:27.

\bibitem{KuttenP98}
S.~Kutten, D.~Peleg, \href{https://doi.org/10.1006/jagm.1998.0929}{Fast distributed construction of small \emph{k}-dominating sets and applications}, J. Algorithms 28~(1) (1998) 40--66.
\newblock \href {https://doi.org/10.1006/jagm.1998.0929} {\path{doi:10.1006/jagm.1998.0929}}.
\newline\urlprefix\url{https://doi.org/10.1006/jagm.1998.0929}

\bibitem{PelegR00}
D.~Peleg, V.~Rubinovich, \href{https://doi.org/10.1137/S0097539700369740}{A near-tight lower bound on the time complexity of distributed minimum-weight spanning tree construction}, {SIAM} J. Comput. 30~(5) (2000) 1427--1442.
\newblock \href {https://doi.org/10.1137/S0097539700369740} {\path{doi:10.1137/S0097539700369740}}.
\newline\urlprefix\url{https://doi.org/10.1137/S0097539700369740}

\bibitem{MMM23}
A.~R. Molla, K.~Mondal, W.~K.~M. Jr., Fast deterministic gathering with detection on arbitrary graphs: The power of many robots, in: {IEEE} International Parallel and Distributed Processing Symposium, {IPDPS}, {IEEE}, 2023, pp. 47--57.

\bibitem{Ta-ShmaZ14}
A.~Ta{-}Shma, U.~Zwick, \href{https://doi.org/10.1145/2601068}{Deterministic rendezvous, treasure hunts, and strongly universal exploration sequences}, {ACM} Trans. Algorithms 10~(3) (2014) 12:1--12:15.
\newblock \href {https://doi.org/10.1145/2601068} {\path{doi:10.1145/2601068}}.
\newline\urlprefix\url{https://doi.org/10.1145/2601068}

\bibitem{PattanayakBCM24}
D.~Pattanayak, S.~Bhagat, S.~G. Chaudhuri, A.~R. Molla, Maximal independet set via mobile agents, in: ICDCN, ACM, 2024, pp. 74--83.

\bibitem{ChandMS23}
P.~K. Chand, A.~R. Molla, S.~Sivasubramaniam, Run for cover: Dominating set via mobile agents, in: ALGOWIN, Springer, 2023, pp. 133--150.

\end{thebibliography}
\end{document}